\newif\ifstoc
\stocfalse

\documentclass[11pt]{article}%
\usepackage{fullpage}
\usepackage{bm}
\usepackage{tikz}
\usepackage{tikz-3dplot}
\usetikzlibrary{shapes,calc,positioning}
\usetikzlibrary{calc,intersections}
\usetikzlibrary{arrows,decorations.markings}
\usepackage[normalem]{ulem}
\usepackage{algpseudocode}

\usepackage{amsmath}
\usepackage{mathtools}
\usepackage{amsfonts}
\usepackage{amssymb}
\usepackage{bbm}
\usepackage{graphicx}
\usepackage[hidelinks]{hyperref}
\usepackage{amsthm}
\usepackage{thmtools}
\usepackage{thm-restate}
\declaretheorem[numberwithin=section]{theorem}
\declaretheorem[numberlike=theorem]{definition}
\declaretheorem[numberlike=theorem]{conjecture}
\declaretheorem[numberlike=theorem]{lemma}
\declaretheorem[numberlike=theorem]{proposition}
\declaretheorem[numberlike=theorem]{example}
\declaretheorem[numberlike=theorem]{corollary}
\declaretheorem[numberlike=theorem]{claim}
\usepackage[capitalize]{cleveref}
\crefname{claim}{claim}{claims}
\Crefname{claim}{Claim}{Claims}

\newcommand{\card}[1]{\lvert#1\rvert}

\newcommand{\suchthat}{\mathrel{:}}

\newcommand{\RR}{\mathbb{R}}

\newcommand{\norm}[1]{{\lVert#1\rVert}}
\newcommand{\norms}[1]{{\lVert#1\rVert}^2}
\newcommand{\floor}[1]{\left\lfloor#1\right\rfloor}
\newcommand{\abs}[1]{\lvert#1\rvert}

\newcommand{\conv}{\operatorname{conv}}

\newcommand{\inner}[2]{#1\cdot #2}
\newcommand{\goesto}{\rightarrow}

\newcommand{\argmin}{\operatorname{argmin}}
\newcommand{\argmax}{\operatorname{argmax}}

\newcommand{\diam}{\operatorname{diam}}

\newcommand{\prob}{\mathbb{P}}
\newcommand{\grad}{\nabla}
\newcommand{\eps}{\epsilon}

\newcommand{\poly}{\operatorname{poly}}



\newcommand{\aff}{\operatorname{aff}}
\newcommand{\linspan}{\operatorname{span}}
\newcommand{\dist}{\operatorname{dist}}

\newcommand{\area}{\operatorname{area}}
\newcommand{\cone}{\operatorname{cone}}
\newcommand{\vf}{\operatorname{vf}}
\newcommand{\facets}{\operatorname{facets}}
\newcommand{\vertices}{\operatorname{vertices}}
\newcommand{\pwidth}{\operatorname{PWidth}}
\newcommand{\pdirw}{\operatorname{PDirW}}
\newcommand{\dirw}{\operatorname{dirW}}
\newcommand{\relint}{\operatorname{relint}}
\newcommand{\minwidth}{\operatorname{minwidth}}
\newcommand{\width}{\operatorname{width}}
\newcommand{\faces}{\operatorname{faces}}
\newcommand{\hfrac}[2]{{#1}/{#2}}
\newcommand{\deq}{\overset{d}{=}}
\newcommand{\sphere}{\mathcal{S}}
\newcommand{\krank}{\operatorname{K-rank}}
\newcommand{\rkrank}[2]{\operatorname{K-rank}_{#1}(#2)}
\newcommand{\gv}{\mathcal G}
\newcommand*{\email}[1]{\href{mailto:#1}{\nolinkurl{#1}} }

\newenvironment{claimproof}{%
\begin{proof}[Proof of claim]%
}{%
\begingroup%
\end{proof}%
\endgroup}

\def\final{1}  

\ifnum\final=0  
\newcommand{\lnote}[1]{[{Luis: \bf #1}]}
\newcommand{\cnote}[1]{[{Chang: \bf #1}]}
\newcommand{\anonnote}[1]{[{anon: \bf #1}]}
\newcommand{\sidecomment}[1]{\marginpar{\tiny #1}}
\newcommand{\details}[1]{{\color{blue}\ [[#1]] }}
\else 
\newcommand{\lnote}[1]{}
\newcommand{\cnote}[1]{}
\newcommand{\anonnote}[1]{}
\newcommand{\sidecomment}[1]{}
\newcommand{\details}[1]{}
\fi  

\title{The smoothed complexity of Frank-Wolfe methods via conditioning of random matrices and polytopes}
\date{}
\author{Luis Rademacher\\
University of California, Davis\\
\email{lrademac@ucdavis.edu}
\and 
Chang Shu\\
University of California, Davis\\
\email{ccshu@ucdavis.edu}
}
\begin{document}

\maketitle

\begin{abstract}
	Frank-Wolfe methods are popular for optimization over a polytope. 
	One of the reasons is because they do not need projection onto the polytope but only linear optimization over it. 
	To understand its complexity, a fruitful approach in many works has been the use of condition measures of polytopes.
	Lacoste-Julien and Jaggi introduced a condition number for polytopes and showed linear convergence for several variations of the method. 
	The actual running time can still be exponential in the worst case (when the condition number is exponential). 
	We study the smoothed complexity of the condition number, namely the condition number of small random perturbations of the input polytope and show that it is polynomial for any simplex and exponential for general polytopes.
	Our results also apply to other condition measures of polytopes that have been proposed for the analysis of Frank-Wolfe methods: vertex-facet distance (Beck and Shtern) and facial distance (Pe\~na and Rodr\'iguez).

	Our argument for polytopes is a refinement of an argument that we develop to study the conditioning of random matrices.
	The basic argument shows that for $c>1$ a $d$-by-$n$ random Gaussian matrix with $n \geq cd$ has a $d$-by-$d$ submatrix with minimum singular value that is exponentially small with high probability.
	This also has consequences on known results about the robust uniqueness of tensor decompositions, the complexity of the simplex method and the diameter of polytopes.

\end{abstract}

\section{Introduction}\label{sec:introduction}

Frank-Wolfe methods (FWMs) \cite{doi:10.1002/nav.3800030109} are a family of algorithms that attempt to minimize a differentiable function over a convex set. 
For concreteness we start by describing the basic Frank-Wolfe method to minimize a differentiable function $f:C \mapsto \RR$ where $C \subseteq \RR^d$ is a compact convex set.
It is an iterative method and proceeds as follows:

	\begin{minipage}{0.6\textwidth}
		\begin{algorithmic}
			\State Let $x_0 \in C$.
			\For{$k=0,\dotsc, K$}
			\State     Compute $y \in \argmin_{x \in C} (\grad f(x_k))^T x$.
			\State     Let $x_{k+1} = x_k + \alpha^*(y-x_k)$, where $\alpha^*$ is a suitable step size.
			\EndFor
		\end{algorithmic}
	\end{minipage}
\begin{minipage}{0.4\textwidth}
\begin{center}
	\begin{tikzpicture}[scale=0.6]
	\node[align=left] at (0,6) {\small $\min_{x \in C} \norm{x}^2$};
	\draw (-4,2) node[anchor=north]{}
	-- (3,1) node[anchor=north]{}
	-- (2.7,5.5) node[anchor=south]{}
	-- (-3,4.7) node[anchor=south]{}
	-- cycle;
	\filldraw[black] (0.2,3.8) circle (1pt) node[anchor=south] {\small $x_0$};
	\filldraw[black] (2,2) circle (1pt) node[anchor=south] {\small $x_1$};
	\filldraw[black] (0,2) circle (1pt) node[anchor=south] {\small $x_2$};
	\filldraw[black] (0,1) circle (1pt) node[anchor=north] {\small $O$};
	\draw[dashed] (-4,2)--(2,2);
	\draw[decoration={markings,mark=at position 1 with
		{\arrow[scale=1.2,>=stealth]{>}}},postaction={decorate}] (0.2,3.8)--(2,2);
	\draw[dashed] (3,1)--(0.2,3.8);
	\draw[decoration={markings,mark=at position 1 with
		{\arrow[scale=1.2,>=stealth]{>}}},postaction={decorate}] (2,2)--(0,2);
	\end{tikzpicture}
\end{center}
\end{minipage}
\smallskip

Some of our results are about Wolfe's method \cite{Wolfe1976}, which is a variation of Frank-Wolfe methods specialized to the minimum norm point problem in a polytope (that is, a bounded convex polyhedron).
 
\subsection{Our contributions and related work}

In this paper we are interested in the complexity of FWMs.
The time complexity of Wolfe's method is know to be exponential in the worst case (by an upper bound in \cite{Wolfe1976} and a lower bound in \cite{DBLP:journals/siamcomp/LoeraHR20}).
There is a large body of work proving linear convergence of several variations of FWMs \cite{MR842638, DBLP:journals/corr/abs-1301-4666,
	LacosteJulien2013AnAI, NIPS2015_5925, Beck2017,doi:10.1137/15M1009937, MR3920711, 
	DBLP:conf/aistats/PedregosaNAJ20}.
We are particularly interested in \cite{LacosteJulien2013AnAI, NIPS2015_5925,Beck2017,doi:10.1137/15M1009937,MR3920711} which prove \emph{global} linear convergence of certain variations of FWMs: F-W with away steps, pairwise F-W and Wolfe's method when the feasible region is a polytope $C = \conv(A)$ for finite $A \subseteq \RR^d$. 
In these results the upper bound on the running time (actual speed of linear convergence) depends on a condition number of $C$. 
Informally speaking, the dependence is of the following kind: if $x_t$ is the current point after $t$ iterations, then the function value satisfies $f(x_t) - f^* \leq (1-\kappa)^t(f(x_0) - f^*)$ where $f^*$ is the optimal value, $x_0$ is the initial point and $0 \leq \kappa \leq 1$ is a measure of conditioning.
If $\kappa$ is small, then convergence is slow.
In the previously mentioned papers, $\kappa$ is of the form $\textrm{``something''}/\diam(C)$, where ``something'' can be:
\begin{itemize}
\item \cite{NIPS2015_5925} minimum width, $\minwidth(A) =\min_{S \subseteq A} \width (S)$ ($\width$ is standard, see \cref{sec:minwidth});
\item \cite{NIPS2015_5925} pyramidal width, $\pwidth (A)$;
\item \cite{Beck2017} vertex-facet distance, $\vf(C) = \min_{F \in \facets(C)} d(\aff F,\vertices(C) \setminus F)$; or
\item \cite{MR3920711} facial distance, $\Phi(C) = \min_{\substack{F \in \faces(C) \\ \emptyset \varsubsetneq F \varsubsetneq C}} d(F, \conv (\vertices (C)\setminus F))$.
\end{itemize}
We do not provide a definition of pyramidal width at this point as it is complicated and it was shown in \cite{MR3920711} that $\pwidth(A) = \Phi(C)$ (\cref{thm:facialdistancepwidth} here).
It is also known that $\minwidth(A) \leq \pwidth(A)$ \cite[Section 3.1]{NIPS2015_5925}.
We start with the observation that $\Phi(C) \leq \vf(C)$ (\cref{thm:summary}).
(Note that the reverse inequality was claimed in \cite{MR3920711}, but the cube $[0,1]^d$ is a counterexample: $\Phi([0,1]^d) =1/\sqrt{d}$ while $\vf([0,1]^d)=1$.)
This implies that all four quantities lie between $\minwidth(A)$ and $\vf(C)$ (\cref{thm:summary}).
It follows from \cite{DBLP:journals/siamcomp/LoeraHR20} that all of them can be exponentially small as a function of the bit-length of $A$.
In fact, a stronger result follows from the work of Alon and Vu \cite{alon1997anti} combined with the stated inequalities. Alon and Vu showed that there is a 0--1 simplex $S$ such that $\vf(S)$ is sub-exponentially small in the dimension (\cref{cor:alonvu}).
The connection between polytope conditioning for FWMs and the Alon and Vu result was observed in \cite{NIPS2015_5925}.

The main contributions of this paper are about the smoothed analysis of FWMs and the condition numbers of matrices and polytopes.
Smoothed analysis \cite{10.1145/380752.380813} is an approach to understand the behavior of algorithms that are efficient in practice but are inefficient in the worst case. 
The main idea is to study small random perturbations of any given instance of a problem.
Suppose that the instance is described by a vector $x \in \RR^n$.
Then one aims to understand $T(x+g)$, where $g \in \RR^n$ is a random vector with distribution $N(0,\sigma^2 \mathbf{I}_n)$ and $T$ is a measure of complexity (for example, $T(x)$ could be the running time of a particular algorithm on input $x$).
We adopt a definition that first appeared in \cite{DBLP:conf/stoc/BeierV04, DBLP:journals/siamcomp/BeierV06}.
\begin{definition}[{\cite{DBLP:conf/ipco/RoglinV05} \cite[Section 1.1]{DBLP:journals/mp/RoglinV07}}]\label{def:psc}
We say $T$ has \emph{(probabilistic) polynomial smoothed complexity} if there is a polynomial $p$ such that 
\[
\max_{x \in \RR^n, \norm{x} \leq 1} \prob_g\bigl( T(x+g) \geq p(n, 1/\sigma, 1/\delta) \bigr) \leq \delta
\]
\end{definition}
Our first smoothed analysis result concerns FWMs minimizing a convex function on a simplex (\cref{sec:simplices}). 
We show that $\minwidth$ has good smoothed complexity (\cref{lem:simplex}). 
This implies the following result on polytope conditioning that can be combined with results in \cite{NIPS2015_5925} to show polynomial smoothed time complexity of several FWMs for the minimization of a convex function in any simplex:
\begin{restatable}{theorem}{thmsimplexsmoothed}\label{thm:simplex smoothed}
	Let $A= \{A_1, \dotsc, A_{d+1}\}$ be a set of independent Gaussian random vectors with means $\mu_i$, $\norm{\mu_i} \leq 1, i \in [d+1]$, and covariance matrix $\sigma^2 \mathbf{I}_d$.
	Then for $\delta > 0$, with probability at least $1-\delta$, the measure of conditioning $\kappa = \frac{\pwidth(A)}{\diam(A)}$ of $A$ is at least some inverse polynomial in $d$, $\hfrac{1}{\sigma}$ and $\hfrac{1}{\delta}$.
\end{restatable}
Note that even the problem of finding the minimum norm point in a simplex is not known to have a simple polynomial time algorithm. 
All polynomial time algorithms we know for such a special case are general purpose convex programming algorithms such as the ellipsoid method.
Moreover, \cite{DBLP:journals/siamcomp/LoeraHR20} shows that the linear programming problem reduces in strongly polynomial time to the minimum norm point in a simplex problem. 
This suggests that to find a simple polynomial time algorithm for the minimum norm point in a simplex is hard and, in particular, to find a strongly polynomial time algorithm would imply the existence of a strongly polynomial time algorithm for linear programming, which would solve a major open problem.

Our second smoothed analysis result concerns condition measures of general polytopes (\cref{sec:polytopeconditioning}).
We show that the standard global linear convergence results for FWMs mentioned above based on polytope conditioning cannot guarantee polynomial complexity for general polytopes in the average or smoothed sense.
More specifically, for V-polytopes $\conv(A)$ with $\card{A}$ and $d$ large and comparable, $d \approx \delta \card{A}$, $\delta \in (0,1)$, we show that vertex-facet distance does not have polynomial smoothed complexity. 
Given that the complexity here increases as $\vf(A)$ gets smaller, in the context of \cref{def:psc} one sets $T=1/\vf$.
It is enough to take $x=0$ there and we show:
\begin{restatable}{theorem}{thmpolytopeconditioning}\label{thm:polytopeconditioning}
	Let $\delta \in (0,1)$.
	Suppose $A = \{A_1, \dotsc, A_{n+1}\}$ is a set of iid.\ standard Gaussian random vectors in $\RR^d$ and $d = \floor{ \delta n }$. 
	Let $P_{n+1} = \conv(A_1, \dotsc, A_{n+1} )$. 
	Then	
	\[
	\prob \bigl(\diam(P_{n+1}) \geq \sqrt{d} \bigr)  \geq 1 - e^{-\frac{nd}{32}},
	\]
	and there exists constants $0<c, c'<1$ (that depends only $\delta$) such that,
	\[
	\lim_{n \goesto \infty} \prob \bigl(\vf(P_{n+1}) \leq c^d  \bigr) \geq c'. 
	\]
	Hence the measure of conditioning $\kappa = \frac{\vf(P_{n+1})}{\diam(P_{n+1})}$ of $A$ is exponentially small in $d$ with constant probability.
\end{restatable}
\Cref{thm:polytopeconditioning} combined with \cref{thm:summary} implies that none of the four measures of polytope conditioning ($\minwidth$, $\pwidth$, $\Phi$, $\vf$) has polynomial smoothed complexity.

A way of interpreting \cref{thm:polytopeconditioning} is that the standard conditioning measures of polytopes for FWMs are somewhat pessimistic and can appear ill-conditioned even then polytope is bad only locally. 
For example, vertex-facet distance can be small even if one vertex and one facet are bad while the rest of the polytope is good.
In other words, it may still be possible to show smoothed polynomial complexity of FWMs in a different way.

\Cref{thm:polytopeconditioning} is a statement about the minimum distance between the affine hull of $d$ \cnote{to confirm} points that form a facet and a vertex not on that facet.
In order to understand this problem we study first a simplified version where we replace affine hull by span and we remove the restriction that the $d-1$ points form a facet.
Namely, we study the following question: given $n$ standard Gaussian random points in $\RR^d$, how close can one of the points be to the span of some $d-1$ others when $n$ is somewhat larger than $d$, say, $n=2d$?
This question is easier to understand than the polytope version and it relates to conditioning of random matrices and the restricted isometry property in compressive sensing.
The relation starts from the known observation (\cref{lem:RudelsonVershynin}) that the minimum point-hyperplane distance is, up to polynomial factors, the same as the smallest singular value of a matrix. 
Given this, our question is essentially equivalent to: given an $d$-by-$n$ random matrix with iid.\ standard Gaussian entries, what is the minimum of the smallest singular values over $d$-by-$d$ submatrices?
We answer this question by showing that when $n/d \geq c > 1$ the minimum smallest singular value above (and, equivalently, minimum point-hyperplane distance) is exponentially small: 
\begin{restatable}{theorem}{thmsigmamgeneral}\label{thm:sigmam general}
	Let $A$ be an $d$-by-$n$ random matrix with iid.\ standard Gaussian entries with $d \geq 2$ and $\frac{n}{d} \geq c_0> 1$.
	Then, there exist constants $c_2, c_4 > 1$, $0< c_6 < 1$ (that depend only on $c_0$) such that with probability at least $1-2c_4 c_6^d$,
	\[
	\min_{S \subseteq [n], \card{S}=d} \sigma_d (A_S) \leq  \frac{1}{c_4c_2^{d-1}}.
	\]
\end{restatable}

\begin{restatable}{theorem}{thmsigmamgeneraltwo}\label{thm:sigmam general 2}
	Let $A$ be an $d$-by-$n$ random matrix with iid.\ standard Gaussian entries with $d \geq 2$ and $1 < \frac{n}{d-1} \leq C_0$.\details{can use $(n-1)/(d-1)\leq C_0$ with similar proof}
	Then, there exist constants $C_1 > 1$, $0 < C_2 < 1$ (that depend only on $ C_0$) such that with probability at least $1-nC_2^{d-1}$,
	\[
	\min_{S \subseteq [n], \card{S}=d} \sigma_d (A_S) \geq \frac{1}{C_1^{d-1}}.
	\]
\end{restatable}

While \cref{thm:sigmam general,thm:sigmam general 2} are new as far as we know, there is a large body of work, partly motivated by compressive sensing, that studies questions related to them.
In that area one is generally interested in showing that all $d$-by-$k$ submatrices of $A$ are well-conditioned, say, $\sigma_1/\sigma_k$ is no more than a constant (the \emph{restricted isometry property} of Cand\`es and Tao \cite{DBLP:journals/tit/CandesT05,DBLP:journals/tit/CandesT06}).
This can only happen when $k$ is much smaller than $d$, a regime very different from our case $k=d$.
The standard analyses in compressive sensing as well as recent results such as \cite{cai2019asymptotic} do not seem to be able to clarify the behavior in our regime.

The idea of the proof of \cref{thm:sigmam general} (\cref{sec:matrixconditioning}) is the following:
Consider the case $n=2d$ for concreteness and aim to show that with constant probability one point is exponentially close to the span of $d-1$ others.
Let $\mathcal{S}$ be the family of sets of $d-1$ columns of $A$.
For $S \in \mathcal{S}$, let $\mathcal{B}_S$ be the set of points in $\RR^d$ within distance $\eps$ of $\linspan{S}$.
Let $V_\eps = \bigcup_{S \in \mathcal S} \mathcal{B}_S$.
It is enough to show that for $\eps = \hfrac{1}{c^d}$, $c>1$, the Gaussian volume $\gv (V_\eps)$ is at least a constant.
We do this by lower bounding it using the \emph{first two terms of the inclusion-exclusion principle (Bonferroni inequality)}:
\[
	\gv(V_\eps) \geq \sum_S \gv (\mathcal{B}_S) - \frac{1}{2} \sum_{S,T: S\neq T} \gv(\mathcal{B}_S \cap \mathcal{B}_T).
\] 
Note that $\mathcal{B}_S \cap \mathcal{B}_T$ can be large if $S$ and $T$ share many columns.
To deal with this difficulty, replace $\mathcal S$ above with a large subfamily $\mathcal{T} \subseteq \mathcal S$ of subsets of columns where each pair of subsets has few columns in common by picking separated subsets \emph{greedily (Gilbert-Varshamov bound)}.
See \cite{razborov1988bounded}, \cite[Lemma 19.3]{MR2865719} for another instance of Bonferroni's inequality with almost pairwise independence.

While \cref{thm:sigmam general,thm:sigmam general 2} are results about random matrices, they have direct implications in the analysis of algorithms: 
In \cref{sec:tensor} we discuss how \cref{thm:sigmam general} condtions the applicability of the robustness of tensor decomposition result by Bhaskara, Charikar and Vijayaraghavan \cite{bhaskara2014uniqueness}.
In \cref{sec:deltadistance} we discuss how \cref{thm:sigmam general} conditions the applicability of results about the complexity of the simplex method and the diameter of polytopes in 
\cite{brunsch2013finding,%
DBLP:conf/stacs/BrunschGR15,%
DBLP:journals/dcg/DadushH16,%
DBLP:journals/mp/EisenbrandV17%
}.

\ifstoc
See \cref{sec:preliminaries} for preliminaires such as basic results and definitions.
\else
\section{Preliminaries}\label{sec:preliminaries}

\subsection{Notations}

For $v \in \RR^d$ and $i \in [d]$, let $v_{-i}$ denote vector $v$ with coordinate $v_i$ removed, that is $v_{-i} := (v_1, \dotsc, v_{i-1},v_{i+1},\dotsc, v_d)$.
If $v \neq 0$, let $\hat{v} := v/\norm{v}_2$.
Let $B(x,\eps) := \{y \in \RR^d \suchthat \norm{y-x} _2\leq \eps\}$.
Let $\sphere^{d-1}$ denote the $(d-1)$-dimensional unit sphere in $\RR^d$.
For $v \in \sphere^{d-1}$, denote the spherical cap centered at $v$ with angle $\alpha$ as $\mathcal{C}_\alpha(v) := \{x \in \sphere^{d-1} \suchthat \inner{v}{x} \geq \cos \alpha\}$.
For $A \subseteq \RR^d$, let $A_\eps := \{x \in \RR^d \suchthat \dist(x,A) \leq \eps \}$, $A_{-\eps} := \{x \in \RR^d \suchthat B(x,\eps) \subset A\}$. 
Let $\mathcal{G}$ denote the standard multivariate Gaussian probability measure.
For random variables or distributions $X,Y$, notation $X \deq Y$ states that $X$ and $Y$ have the same distribution.

\subsection{Noncentral Chi-Square Distribution}
We first recall the definition of noncentral chi-square distribution.
\begin{definition}\label{def:noncentral chi-square}
	Let $Y_1, Y_2, \dotsc, Y_k$ be independent Gaussian random variables with means $\mu_i$ and unit variance.
	Then the random variable $\sum_{i=1}^k Y_i^2$ is distributed according to the noncentral chi-square distribution with $k$ degrees of freedom and noncentrality parameter $\lambda = \sum_{i=1}^k \mu_i^2$. 
	The probability density function of noncentral chi-square distribution is given by
	\[
		f_Y(x;k,\lambda) = \sum_{i=0}^\infty \frac{e^{-\lambda /2}(\lambda /2)^i}{i!}f_{Z_{k+2i}}(x),
	\]
	where $Z_q$ is distributed as (central) chi-square with $q$ degrees of freedom, denoted as $\chi^2_q$.
\end{definition}

We show an anti-concentration inequality of noncentral chi-square distribution by comparing to (central) chi-square distribution.
\begin{lemma}\label{lem:noncentral chi-square}
	Let $\mu \in \RR$.
    Let $X_i \sim \mathcal{N}(0,\sigma^2)$, $Y_i \sim \mathcal{N}(\mu,\sigma^2)$, $i \in \left[k\right]$, be independent.
	Then 
    \[ 
	\mathbb{P}\left(\sum_{i=1}^{k}X_i^2 \geq t^2 \right) \leq	\mathbb{P} \left( \sum_{i=1}^{k}Y_i^2 \geq t^2 \right).
	\] 
\end{lemma}
\begin{proof}
	The proof follows directly from the density function of noncentral chi-square distribution and some basic facts about the chi-square distribution.
	We consider random variables $Y_i/\sigma \sim \mathcal{N}(\mu/ \sigma,1)$.
	Then $\sum_{i=1}^{k} Y_i^2/ \sigma ^2$ is distributed as noncentral chi-square with $k$ degrees of freedom and noncentrality parameter $k\mu^2  /\sigma^2$. 
	From \cref{def:noncentral chi-square},
	\begin{align*}
		\mathbb{P}\left(\sum_{i=1}^k \frac{Y_i^2}{\sigma^2} \ge \frac{t^2}{\sigma^2}\right) &= \sum_{i=0}^\infty \frac{e^{-k\mu^2  /2\sigma^2}(k\mu^2  /2\sigma^2)^i}{i!}\mathbb{P} \left( Z_{k+2i} \ge \frac{t^2}{\sigma^2} \right)\\
		&\ge \sum_{i=0}^\infty \frac{e^{-k\mu^2  /2\sigma^2}(k\mu^2  /2\sigma^2)^i}{i!}\mathbb{P} \left(Z_{k} \ge \frac{t^2}{\sigma^2} \right)\\
		& = \mathbb{P} \left(Z_{k} \ge \frac{t^2}{\sigma^2} \right) \\
		& = \mathbb{P} \left(\sum_{i=1}^{k}X_i^2 \geq t^2 \right),
	\end{align*}
	where $Z_k \sim \chi^2_{k}$. The inequality comes from the fact that chi-square random variable with $(k+i)$ degrees of freedom is equal in distribution to the sum of a chi-square random variable with $k$ degrees of freedom and squares of $i$ independent standard Gaussian random variables, so that $Z_{k+2i}$ has larger tail than $Z_k$.
\end{proof}

The following lemma provides a comparison inequality between the ratio of noncentral chi-square random variables and chi-square random variables, which is used in the proof of \cref{lem:twobandsaff}.

\begin{lemma}\label{lem: monotonicity}
	Let $\mu \in \RR$.
	Let $X_0, X_i, Y_0 \sim \mathcal{N}(0,\sigma^2)$, $Y_i \sim \mathcal{N}(\mu,\sigma^2)$, $i \in \left[k\right]$ and be independent.
	Then for any $t \in (0,1)$,
	\[ 
	\mathbb{P}\left(\frac{Y_0^2}{Y_0^2 + \sum_{1}^{n} Y_i^2} \geq t^2\right) \leq	\mathbb{P}\left(\frac{X_0^2}{X_0^2 + \sum_{1}^{n} X_i^2} \geq t^2\right).
	\] 
\end{lemma}
\begin{proof}
	Let $f$ denote the probability density function of $X_0^2$ and $Y_0^2$.
	By the law of total expectation,
	\begin{align*}
	\mathbb{P}\left( \frac{Y_0^2}{Y_0^2 + \sum_{i=1}^{n} Y_i^2} \geq t^2 \right) 
	& = \int_{0}^{\infty} \mathbb{P}\left(\frac{y}{y + \sum_{i=1}^{n} Y_i^2} \geq t^2\right) f(y) dy\\
	& =  \int_{0}^{\infty} \mathbb{P}\left(\sum_{i=1}^{n} Y_i^2 \leq (1/t^2-1) y\right) f(y) dy\\
	& \leq \int_{0}^{\infty} \mathbb{P}\left(\sum_{i=1}^{n} X_i^2 \leq (1/t^2-1) x \right) f(x)dx \quad \text{(\cref{lem:noncentral chi-square})} \\
	& = \int_{0}^{\infty} \mathbb{P}\left(\frac{x}{x + \sum_{i=1}^{n} X_i^2} \geq t^2 \right) f(x) dx\\
	& = \mathbb{P}\left(\frac{X_0^2}{X_0^2 + \sum_{i=1}^{n} X_i^2} \geq t^2 \right).
	\end{align*}
\end{proof}

\subsection{Concentration and tail inequalities}


\begin{lemma}\label{lem: gaussian conc}
	Let $X \sim \mathcal{N}(0,\sigma^2)$, $Y \sim \mathcal{N}(c,\sigma^2)$, $c \in \RR$, then for any $t>0$,
	\[
	\mathbb{P}(\norm{Y}_2 < t) \leq \mathbb{P}(\norm{X}_2 < t).
	\]
\end{lemma}
\begin{proof}
	Without loss of generality, we may assume $c>0$.
	Then
	\begin{align*}
		\mathbb{P}(\norm{Y}_2 < t) &=\int_{-t}^{t}e^{\frac{-(x-c)^2}{2\sigma^2}}dx\\
		&=\int_{-t-c}^{t-c}e^{\frac{-x^2}{2\sigma^2}}dx\\
		&=\int_{-t-c}^{-t}e^{\frac{-x^2}{2\sigma^2}}dx + \int_{-t}^{t}e^{\frac{-x^2}{2\sigma^2}}dx - \int_{-t-c}^{t}e^{\frac{-x^2}{2\sigma^2}}dx\\
		& = \mathbb{P}\bigl(\norm{X}_2 < t \bigr) + \left(\int_{t}^{t+c}e^{\frac{-x^2}{2\sigma^2}}dx-\int_{t-c}^{t}e^{\frac{-x^2}{2\sigma^2}}dx \right)\\
		&\leq \mathbb{P}\bigl(\norm{X}_2 < t\bigr).
	\end{align*}
\end{proof}

\begin{lemma}[\cite{laurent2000}]\label{lem:laurent}
	Let $(X_1, \dotsc, X_n)$ be iid.\ standard Gaussian variables. Let $\alpha_1, \dotsc, \alpha_n$ be nonnegative. Let $Z= \sum_{i=1}^{n} \alpha_i(X_i^2-1)$. Then, the following inequalities hold for any positive $t$:
	\begin{align*}
		\mathbb{P} \bigl( Z \ge 2\norm{\alpha}_2\sqrt{t}+2\norm{\alpha}_{\infty}t \bigr) &\le \exp(-t),\\
		\mathbb{P} \bigl( Z \le -2\norm{\alpha}_2\sqrt{t} \bigr) &\le \exp(-t).
	\end{align*}
\end{lemma}
\details{
	Let $(X_1, \dotsc, X_d)$ be iid.\ standard Gaussian variables. 
	Let $X= \sum_{i=1}^{d} X_i^2$. Then, the following inequalities hold for any positive $t$:
	\begin{align*}
		\mathbb{P} \bigl( X \le d -2\sqrt{dt} \bigr) &\le \exp(-t).
	\end{align*}

}

\subsection{Gilbert-Varshamov bound}

\begin{lemma}\label{lem:gilbert general}
	Let $A(n,t,w)$ be the maximum number of binary $n$-vectors with exactly $w$ ones and pairwise Hamming distance greater than or equal to $t$.
	Then for any $c_0>1$, there exist constants $c_1 > 0$ and $c_2 > 1$ (that depend only on $c_0$) such that for all $d \geq 1$ and $n/d \geq c_0$ we have $A(n, c_1 d, d) \geq c_2^d$.
\end{lemma}
\begin{proof}
	Pick vectors greedily (Gilbert-Varshamov bound) to get, for integral $t$:
	\[
	A(n,t,w) \geq \frac{\binom{n}{w}}{B(n,t)},
	\]
	where $B(n,t)$ is the number of binary vectors at Hamming distance less than or equal to $t$ from the zero vector.
	We have $B(n,t) = \sum_{k=0}^{t} \binom{n}{k} \leq \left(\frac{ne}{t}\right)^t$ (see footnote\footnote{
	\(\sum_{k=0}^{t} \binom{n}{k} \leq \sum_{k=0}^{t} \frac{n^k}{k!} = \sum_{k=0}^{t} \frac{t^k}{k!}(n/t)^k \leq e^t(n/t)^t \).
}).
	Note that this last inequality is valid also for fractional $0<t\leq n$ using the fact that $(a/b)^b$ is increasing in $b$ for $a>0$, $0 < b \leq a/e$.\details{its derivative is $(a/b)^b(\log(a/b)-1)$}
	Thus, for any $0< c < 1$ and $d \geq 1$ we have\details{using $\binom{n}{k}\geq (n/k)^k$ for $k\geq 1$}
	\begin{align}\label{eqn:gilbert}
	A(n,cd,d) \geq 
	\frac{\binom{n}{d}}{ \left(\frac{ne}{cd}\right)^{cd} } 
	\geq \frac{ (n/d)^d }{ \left(\frac{ne}{cd}\right)^{cd} }
	= \frac{ (n/d)^{d(1-c)} }{ \left(\hfrac{e}{c}\right)^{cd} }
	\geq \frac{ c_0^{d(1-c)} }{ \left(\hfrac{e}{c}\right)^{cd} }
	= \left( \frac{ c_0 }{ \left(\hfrac{c_0e}{c}\right)^{c} } \right)^d
	\end{align}
	We have $\lim_{c \to 0^+} (a/c)^{c} = 1$ and $(a/c)^{c}$ is increasing again for $a>0$, $0 \leq c \leq a/e$.
	Given this, choose $c_1 \in (0,1)$ such that $(c_0 e/c_1)^{c_1} < c_0$.
	Let $c_2 := \frac{ c_0 }{ \left(\hfrac{c_0e}{c_1}\right)^{c_1} } > 1$.
	We have $A(n,c_1 d, d) \geq c_2^d$ for $d\geq 1$.
	The claim follows.
\end{proof}

\subsection{Generalization of Archimedes' formula}
\begin{lemma}\label{lem:archimedes}
	Let $d \geq 3$.
	Let $U$ be a uniformly random $d$-dimensional unit vector.
	Then $(U_1, \dotsc, U_{d-2})$ is uniform in $B^{d-2}$ and $\Pr(\norm{(U_1, \dotsc, U_{d-2})} \leq t) = t^{d-2}$.
\end{lemma}
\begin{proof}
	The first part is well-known, a proof can be found in \cite[Corollary 4]{barthe2005probabilistic}. The second part follows immediately from the first part.
	\details{ \href{https://mathoverflow.net/questions/33129/intuitive-proof-that-the-first-n-2-coordinates-on-a-sphere-are-uniform-in-a}{link}
	}
\end{proof}

\subsection{One-off-distance vs sigma min}
\begin{lemma}[{see e.g. \cite[Lemma 3.5]{DBLP:journals/corr/BhaskaraCMV13} for a proof}]\label{lem:RudelsonVershynin}
	If $A \in \RR^{m \times n}$  has columns $a_1, \dotsc, a_n$ and $m\geq n$, then denoting $a_{-i} = \linspan{(a_j: j \neq i)}$, we have 
	\begin{align*}
	\frac{1}{\sqrt{n}} \min_{i \in [n]}\dist(a_i, a_{-i}) 
	\leq \sigma_{n}(A) 
	\leq \min_{i \in [n]}\dist(a_i, a_{-i}).
	\end{align*}
\end{lemma}

\subsection{Facts about Gaussian random polytopes}

\subsubsection{Gaussian \texorpdfstring{$\eps$}{epsilon}-neighborhood}

\begin{corollary}\label{lem: gauss neighborhood}
	Let $Q$ be a convex set in $\mathbb{R}^d$. 
	Then there exists an absolute constant $c>0$ such that 
	\(
	\mathcal{G}\left(Q \setminus Q_{-\eps} \right)\leq c \eps d^{1/4}
	\).
\end{corollary}

\begin{proof}
	Follows immediately from \cite[Lemma A.2]{chernozhukov2017central} and the fact $\norm{I}_{HS} = \sqrt{d}$ (Hilbert-Schmidt norm).
	Their proof is based on \cite{Ball1993,nazarov2003maximal}.
\end{proof}

\subsubsection{Distances of facets}
\begin{lemma}[{\cite[Theorem 4.4.5]{vershyninbook}}] \label{lem: R-V sigma max}
	Let $X$ be an $m \times n$ random matrix whose entries are iid.\ standard Gaussian random variables.
	Then for $t > 0$,
	\begin{align*}
	\mathbb{P} \left( \sigma_{\max}(X) > c(\sqrt{m} + \sqrt{n} + t) \right) \leq 2e^{-t^2},
	\end{align*}
	where $c$ is some absolute positive constant.
\end{lemma}

\begin{lemma}\label{lem:affmin}
	Given linearly independent vectors $p_1, p_2, \dotsc, p_d \in \mathbb{R}^d$, the shortest vector in their affine hull is
	$v = \hfrac{P^{-1}\mathbbm{1}}{\norms{P^{-1} \mathbbm{1}}}$,
	where  $P = ( p_1 \cdots p_d )^\top$.
	In particular, $\norm{v}= \hfrac{1}{\norm{P^{-1}\mathbbm{1}}}$.
\end{lemma}
\begin{proof}
	From \cite[Lemma 1.2]{DBLP:journals/siamcomp/LoeraHR20}, the shortest vector in the affine hull, $v$, satisfies
	\(	 
		Pv = \norm{v}^2 \mathbbm{1}.
	\)
	Since $P$ has independent columns,
	\(
	v = \norm{v}^2 P^{-1} \mathbbm{1}.
	\)
	Compute norms of vectors in the above equation to get
	\(
	\norm{v} = 1/\norm{P^{-1} \mathbbm{1}}.
	\)
	The claim follows.
\end{proof}

The following lemma can directly generalize to Gaussian random vectors with mean zero and covariance matrix $\sigma^2 \mathbf{I}_{d}$ by scaling by $\sigma$.
\begin{lemma}\label{lem:boundofaffmin}
	Let $X_1,\dots,X_n$ be iid.\ standard Gaussian random vectors in $\RR^d$. 
	For $S \subseteq \left[n\right]$, $|S| = d$, define $V_S$ as the shortest vector in $\aff (X_S)$.
	Then there exists a constant $c > 0$ such that
	\[
	\mathbb{P}\left(\max_{S \subseteq \left[n\right], |S| = d} \norm{V_S} \leq c  (2+\sqrt{n/d}) \right) \geq 1 - 2e^{- d}.
	\]
\end{lemma}
\begin{proof}
	Let $X$ be the matrix whose column vectors are $X_1, \dotsc, X_n$.
	For any $S \subseteq \left[n\right], |S|= d$, $X_S$ is linearly independent with probability 1.
	By \cref*{lem:affmin},
	\begin{align}\label{eqn:shortest vec}
			\|V_S\| = \frac{1}{\|X_S^{-1}\mathbbm{1}\|} \le \frac{1}{\sqrt{d}\  \sigma_{\min}(X_S^{-1})} = \frac{\sigma_{\max}(A_S)}{\sqrt{d}} \le \frac{\sigma_{\max}(A)}{\sqrt{d}}.
	\end{align}
	From \cref{lem: R-V sigma max} we know
	\(
	\mathbb{P} \left( \sigma_{\max}(A) > c (\sqrt{d}+\sqrt{n} + t) \right) \leq 2e^{- t^2}.
	\)
	The claim follows by letting $t = \sqrt{d}$ and applying \eqref{eqn:shortest vec}.
\end{proof}

%
%

\subsubsection{Number of facets}

We will need the fact that the number of facets of the convex hull of $n$ Gaussian random points in $\RR^d$ is exponential in $d$ with high probability when $n=cd$, $c>1$.
We could not find such a result in the literature and we do not see how to deduce it from results on the asymptotic number of facets in stochastic geometry \cite{MR258089,DBLP:journals/dcg/AffentrangerW91, MR2093024, MR2144555, boroczky2018facets} (the difficulties are: either they only determine the expectation or variance of the number of facets, or the bounds are as $n$ goes to infinity for fixed $d$).
Nevertheless, it is easy to deduce what we want from the work of Donoho and Tanner on compressive sensing and the neighborliness of random polytopes. 
We build on top of basic polytope theory from \cite{Ziegler}.

\begin{definition}[Neighborliness]
	A polytope $P$ is $k$-neighborly if every subset of $k$ vertices forms a $(k-1)$-face.
\end{definition}

Let $f_l(P)$ denote the number of $l$-faces of polytope $P$.


\begin{theorem}[\cite{Donoho9452}, Corollary 1.1, Lemma 3.2]\label{thm:neighborliness}
	There exists a function (threshold) $\rho(\delta): (0,1) \to \RR$, $\rho(\delta)  > 0$ with the following property:
	Let $\delta \in (0,1)$.
	Let $d = \floor{\delta n}$.
	Let $\rho < \rho(\delta)$. 
	Let $X_1, \dotsc, X_n$ be iid.\ samples from a Gaussian distribution in $\RR^d$ with non-singular covariance.
	Let $P = \conv\{ X_1, \dotsc, X_n\}$.
	Then
	$
	\lim_{n \goesto \infty} \prob \bigl(\text{$f_1(P) = n$ and $P$ is $\floor{\rho d}$-neighborly} \bigr) = 1 
	$.
\end{theorem}

The above theorem demonstrates, given its assumptions, that when $n$ is large enough, $P$ has $\binom{n}{\floor{\rho d}}$ many $\floor{\rho d}$-faces with high probability.
Note also that $P$ is simplicial (every facet is a simplex) a.s.
Thus, a.s. each facet of $P$ provides at most $\binom{d}{\floor{\rho d}}$ many $\floor{\rho d}$-faces, and the number of facets is at least
\begin{align*}
\frac{\binom{n}{\floor{\rho d}}}{\binom{d}{\floor{\rho d}}} 
\ge \left( \frac{n}{d} \right)^{\floor{\rho d}} 
\ge \left(\frac{1}{\delta}\right)^{\floor{\rho d}} 
\geq c^d,
\end{align*}
for some $c>1$ (and $d$ large enough).
We conclude:
\begin{corollary}\label{lem: num of facets}
	Let $\delta \in (0,1)$.
	Let $P$ be the convex hull of $n$ iid.\ standard Gaussian random points in $\mathbb{R}^d$, $d = \floor{\delta n}$.
	Then there exists a constant $c > 1$ (that depends only on $\delta$) such that
	$
	\lim_{n \goesto \infty} \mathbb{P}\bigl( f_{d}(P) \geq c^d \bigr) = 1
	$.
\end{corollary}
\Cref{lem: num of facets} can probably also be proven directly from different but related neighborliness results by Vershik and Sporyshev \cite{MR1166627}, \cite[Theorem 2]{Donoho9452}.

\subsection{Condition measures of polytopes}

\subsubsection{Width and minwidth}\label{sec:minwidth}
\begin{definition}[Directional width and width]
	The \textit{directional width} of a set $A \subseteq \RR^d$ with respect to a direction $r \in \RR^d$ is defined as $\dirw(A,r):= \sup_{s,v \in \mathcal{A}} \langle\frac{r}{\|r\|},s-v \rangle$. 
	The \textit{width} of $A$, denoted $\width(A)$ is the infimum of the directional width over all directions on its affine hull. 
	
\end{definition}

\begin{definition}[{Minwidth, \cite[Section 3.1]{NIPS2015_5925}}]
	The minwidth of a finite set $A \subseteq \RR^d$, denoted $\minwidth(A)$, is the minimum width over all subsets of $A$.
\end{definition}

\subsubsection{Pyramidal width}
\begin{definition}[Pyramidal directional width, \cite{NIPS2015_5925}]\label{def:pdirw}
	We define the \textit{pyramidal directional width} of a finite set $A \subseteq \RR^d$ with respect to a direction $r \in \RR^d$ and a base point $x\in \conv(A)$ to be
	\[
	\pdirw(A,r,x):= \min_{S \in S_x} \dirw(S \cup \{ s (A,r) \},r) =  \min_{S \in S_x} \max_{s \in A , v \in S} \left\langle\frac{r}{\norm{r}},s-v \right\rangle,
	\]
	where $S_x :=\{T \subseteq A \suchthat x \text{ is a proper convex combination of all the elements in }T \}$ and $s(A,r):= \argmax_{v\in A}\langle r,v \rangle$.
\end{definition}

\begin{definition}[Feasible direction, \cite{NIPS2015_5925}]
	A direction $r$ is feasible for $A$ from $x$ if it points inwards $conv(A)$, i.e. $r \in cone(A-x)$. A direction $r$ is feasible for $A$ if it is feasible for $A$ from some $x \in A$.
\end{definition}

\begin{definition}[Pyramidal width, \cite{NIPS2015_5925}]\label{def:pwidth}
	We define the \textit{pyramidal width} of a finite set $A \subseteq \RR^d$ to be the smallest pyramidal directional width of all its faces,
	\[
	\pwidth(A) := 
	\min_{\substack{K \in \faces(\conv(A))\\ x \in K \\ r \in \cone(K-x) \setminus \{0\} }} \pdirw(K \cap A,r,x).
	\]
\end{definition}

\subsubsection{Vertex-facet distance}

The \emph{vertex-facet distance} polytope conditioning parameter for the analysis of FWMs was introduced in \cite{Beck2017}.
We adopt here the slightly specialized definition in \cite{MR3920711}, which is defined as a property of a polytope independent of the representation, while the original version in \cite{Beck2017} can depend on the numbers used to represent a polytope.
\begin{definition}[vertex-facet distance \cite{Beck2017, MR3920711}]
Let $P \subseteq \RR^d$ be a polytope with  $\dim(\aff(P)) \geq 1$. 
The vertex-facet distance of $P$ is
\[
\vf(P) := \min_{F \in \facets (P)} \dist(\aff (F), \vertices(P) \setminus F).
\]
\end{definition}
\cnote{add examples?}

\subsubsection{Relation between vertex-facet distance and pyramidal width}

We show $\vf \bigl(\conv(A) \bigr) \geq \pwidth(A)$. It seems that this result may have already been know to \cite[comment before Theorem 1, combined with Theorem 2]{MR3920711}, but it is claimed there in the wrong direction.
That direction is impossible as the example of a unit cube shows: $\pwidth([0,1]^d) = 1/\sqrt{d}$ \cite[Lemma 4]{NIPS2015_5925}, but $\vf([0,1]^d) = 1$.

\begin{proposition}\label{prop:vfpwidth}
	Let $A \subseteq \RR^d$ be a finite set with at least two points. 
	Then $\vf \bigl(\conv(A) \bigr) \geq \pwidth(A)$.
\end{proposition}
\begin{figure}[ht]
	\centering
	\includegraphics[width=0.45\columnwidth]{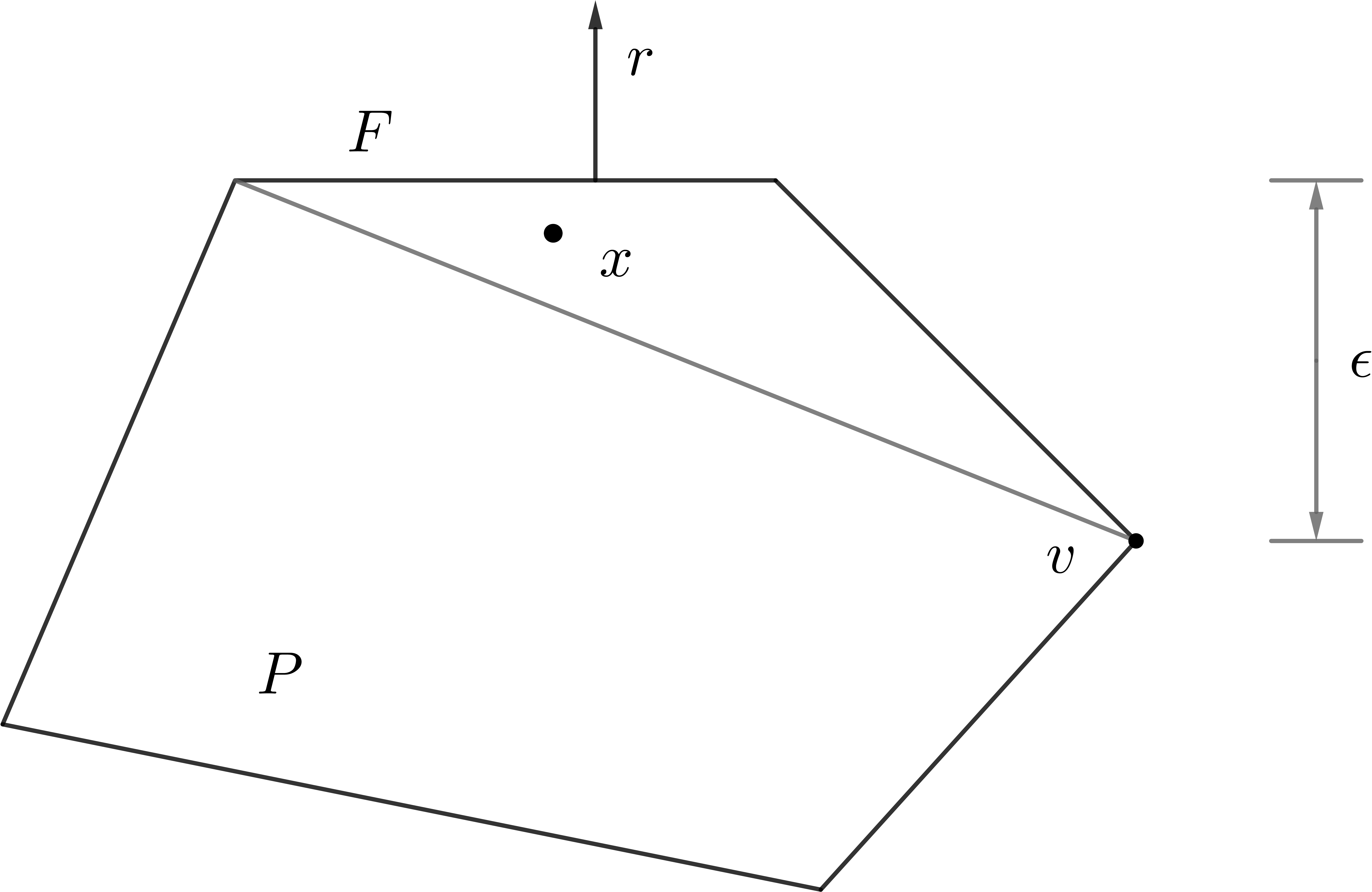}
	\caption{Proof of \cref{prop:vfpwidth}}
	\label{fig:vf vs pwidth}
\end{figure}
\begin{proof}
	Let $P = \conv(A)$.
	Let $F$ be a facet of $P$ and pick $v \in \vertices(P) \setminus F$ so that $\dist \bigl(v, \aff(F)\bigr) = \eps := \vf(P)$.
	Pick $x \in \relint \bigl( \conv(F \cup \{v\}) \bigr) $ and let $r$ be the unit outer normal vector to $F$ (in $\aff(P)$ if $P$ is not full-dimensional).
	We set $K=P$ in \cref{def:pwidth} so that $r \in \cone(K-x)=\aff(P)$ and $\pwidth(A) \leq \pdirw(K\cap A,r,x) = \pdirw(A,r,x)$.
	Now, set $S = A \cap (F \cup \{v\})$ in \cref{def:pdirw} so that, with these choices, $\pdirw(A,r,x) \leq \dirw(S,r) \leq \eps$.
	The claim follows.
\end{proof}

\subsubsection{Facial distance}

\begin{definition}[\cite{MR3920711}]
Let $C \subseteq \RR^d$ be a polytope with $\dim(\aff(C)) \geq 1$. 
The \emph{facial distance} of $C$ is 
\[
\Phi(C) := \min_{\substack{F \in \faces(C) \\ \emptyset \varsubsetneq F \varsubsetneq C}} d(F, \conv (\vertices (C)\setminus F)).
\]
\end{definition}


	
\subsubsection{Relation between facial distance and pyramidal width}

One of the motivations of \cite{MR3920711} to introduce parameter $\Phi$ is that it is the same as $\pwidth$ (except in degenerate cases) while the definition of $\Phi$ is simpler to use in many cases.
We quote their result next.
\begin{theorem}[{\cite[Theorem 2]{MR3920711}}]\label{thm:facialdistancepwidth}
Let $A \subseteq \RR^d$ be a finite set with at least two points.
Then $\Phi\bigl(\conv(A)\bigr) = \pwidth(A)$.
\end{theorem}

\subsubsection{Summary result}
\begin{theorem}\label{thm:summary}
Let $A \subseteq \RR^d$ be a finite set with at least two points. Then
\[
\minwidth(A) \leq \Phi \bigl(\conv(A)\bigr) = \pwidth(A) \leq \vf \bigl(\conv (A) \bigr).
\]
\begin{proof}
	Immediate from \cite[Section 3.1]{NIPS2015_5925}, \cref{thm:facialdistancepwidth,prop:vfpwidth}.
\end{proof}
\end{theorem}

\fi

\section{Conditioning of simplices}\label{sec:simplices}

In this section we show that the smoothed conditioning of any simplex is polynomial.
This implies that several FWMs have smoothed polynomial complexity on the minimum norm point in a simplex problem and the minimization of many convex functions on a simplex.
To put this result in context, we first argue (based on known results) that even a simplex with vertices having 0--1 coordinates can have bad conditioning.
Another relevant context to keep in mind is the fact that linear programming reduces in strongly polynomial time to the minimum norm point in a simplex \cite{DBLP:journals/siamcomp/LoeraHR20}.

\subsection{Equality of \texorpdfstring{$\width$}{width} and \texorpdfstring{$\minwidth$}{minwidth} of a simplex}
We start with the observation that the minwidth of a simplex is the same as its width.
\begin{lemma}\label{lem:minwidth}
	Let $A$ be the vertex set of a simplex in $\RR^d$ and $A_0 \subset A$ which includes more than one vertex.
	Then  $\width(A) \leq\width(A_0)$. 
	In particular, $\minwidth(A) = \width(A)$.
\end{lemma}
\begin{proof}
	We prove by induction in $d$.
    The width of a polytope is the minimum distance between parallel supporting hyperplanes in its affine hull. 
	Width of a $2$-simplex is the minimum height of  triangle, which is smaller than the length of any edge.
	For a $k$-simplex $A$, suppose the width of one of its facet is given by the distance between two parallel $(k-2)$-dimensional planes, $p^{k-2}_1$ and $p^{k-2}_2$.
	One can extend $p^{k-2}_1$ and $p^{k-2}_2$ to parallel hyperplanes in $\mathbb{R}^k$ that enclose $A$.
	Suppose extensions $p^{k-1}_1$ and $p^{k-1}_2$ give the minimum distance.
	Then,
	\[
	\dist\bigl( p^{k-1}_1, p^{k-1}_2 \bigr) 
	= \min_{a \in p^{k-1}_1, b \in p^{k-1}_2} \norm{a-b} 
	\leq \min_{a \in p^{k-2}_1, b \in p^{k-2}_2} \norm{a-b} 
	= \dist\bigl( p^{k-2}_1, p^{k-2}_2 \bigr)
	\]
	which shows that the width of a $k$-simplex is less than the width of any of its facets.
	The claim then follows by induction.
\end{proof}

\subsection{Bad worst case conditioning of a 0--1 simplex}

Lacoste-Julien and Jaggi \cite{NIPS2015_5925} observed that the $\minwidth$ of the unit cube in $\RR^d$ is exponentially small in $d$. 
This example was one of their motivations for introducing $\pwidth$, which is $1/\sqrt{d}$ for the cube. 
Their observation is based on the following result by Alon and Vu:
\begin{theorem}[{\cite[Theorem 3.2.2]{alon1997anti}, \cite[Corollary 27]{MR1785291}}]\label{tmh:alonvu}
There are $d+1$ vectors in $\{0,1\}^d$ that form the vertices of a $d$-dimensional simplex $S$ so that 
	\[
		\frac{2^{d-1}}{d^{d/2}} \leq \vf(S) \leq \frac{2^{d(2+ o(1))}}{d^{d/2}}.
	\]
\end{theorem}
\cite{DBLP:journals/siamcomp/LoeraHR20} observed that $\pwidth$ can be exponentially small in the size (bitlength) of a set of points with integer coordinates.
Using \cref{tmh:alonvu} and the relationships between polytope condition measures, we can immediately strengthen this result and show that this is not just a ``large numbers'' phenomenon, namely, all condition measures are exponentially small even for a 0--1 simplex:
\begin{corollary}\label{coro:minwidth-pw-facial d-vf}
	\label{cor:alonvu}
	There are $d+1$ vectors in $\{0,1\}^d$ that form the vertices of a $d$-dimensional simplex $S$ so that $\width(\vertices(S)) = \minwidth(\vertices(S)) \leq \pwidth(\vertices(S)) = \Phi(S) \leq \vf(S) \leq \frac{2^{d(2+ o(1))}}{d^{d/2}}$.
\end{corollary}
\begin{proof}
	Let $S$ be the $d$-dimensional simplex given by \cref{tmh:alonvu}.
	\Cref{lem:minwidth} gives the leftmost equality.
	The rightmost inequality is one of the conclusions of \cref{tmh:alonvu}.
	The other relations follow from \cref{thm:summary}.
\end{proof}

\subsection{Polynomial smoothed complexity of FWMs on a simplex}
	Now we start analyzing smoothed complexity of FWMs on the minimization of a strongly convex function with Lipschitz gradient on a simplex.
	\begin{definition}
		A differentiable function $f$ is said to have $L$-Lipschitz gradient if for some $L>0$ and for all $ x,y$ in its domain we have
		$
		\norm{\grad f(x) - \grad f(y)} \leq L \norm{x-y}
		$.
	\end{definition}
	
	\begin{definition}
		A differentiable function $f$ is $\mu$-strongly convex if for some $\mu > 0$ and for all $ x,y$ in its domain, we have
		\[
		f(y) \geq f(x) + \grad f(x)^T(y-x) + \frac{\mu}{2} \norm{y-x}^2.
		\]
	\end{definition}
	In \cite[Theorem 1]{NIPS2015_5925}, Lacoste-Julien and Jaggi proved the global linear convergence of FWMs on the minimization of a strongly convex function with Lipschitz gradient:
	suppose $u_t$ is the current point after $t$ good iterations\footnote{The number of good iterations depends on variants of FWMs being used. 
	It is always lower bounded by some linear function of the actual number of iterations. See details in \cite[Theorem 1]{NIPS2015_5925}.}, $f(u_t)$ satisfies
	\begin{align}\label{eqn:conv rate}
			f(u_t) - f^* \leq \left(1-\frac{\mu}{4L}\left(\frac{\pwidth(A)}{\diam(A)}\right)^2\right)^t(f(u_0) - f^*) ,
	\end{align}
	where $f^*$ is the optimal value and $u_0$ is the initial point.
	To show polynomial smoothed complexity, we need to prove that the measure of conditioning $\kappa = \frac{\pwidth(A)}{\diam(A)}$ is at least inverse polynomial \cnote{is it clear that inverse polynomial means $1/p$ but not $p^{-1}$?} in $d, \hfrac{1}{\sigma}, \hfrac{1}{\delta}$.
	We are going to get this by giving a polynomial lower bound on $\pwidth(A)$ and a polynomial upper bound on $\diam(A)$.

\subsubsection{Inverse polynomial smoothed minwidth}

We know from \cref{thm:summary} that $\minwidth \leq  \pwidth$, and from \cref{lem:minwidth} that $\minwidth = \width$ for any simplex. 
Thus, we instead find a lower bound on $\width$, namely the diameter of a ball contained in the simplex, which is also a lower bound on $\pwidth$.
In the next lemma, we prove that a random simplex contains a ball of radius $\Omega(d^{-2})$ with probability close to 1.
\begin{lemma}
\label{lem:simplex}
Let $A =\{A_1, \dotsc, A_{d+1}\}$ be a set of independent Gaussian random vectors with means $\mu_i$, $\norm{\mu_i} \leq 1, i \in [d+1]$, and covariance matrix $\sigma ^ 2 \mathbf{I}_d$.
Then for $\delta > 0$,
\[
\mathbb{P} \left( \minwidth \bigl(\conv (A) \bigr) \geq  \sqrt{2 \pi}\sigma \delta  (d+1)^{-2} \right) \geq 1- \delta .
\]
Moreover,
\begin{align*}
\mathbb{P} \left( \pwidth \bigl(\conv (A) \bigr) \geq  \sqrt{2 \pi}\sigma \delta  (d+1)^{-2} \right) \geq 1- \delta .
\end{align*}

\end{lemma}
\begin{proof}
	It is easy to see that $A$ forms a simplex with probability 1. 
	From \cref{lem:minwidth}, we know the minwidth of a simplex is its width.
	Let $D_i$ be the distance from $A_i$ to the affine hull of its opposite facet, $\aff \{A_j \suchthat  j\neq i \}$.
	Conditioning on $\aff \{A_j \suchthat  j\neq i \}$, by the rotational invariance of Gaussian distribution, $D_i$ is equal in distribution to the absolute value of a Gaussian random variable with mean $\mu \in \RR $ (not necessarily be zero) and variance $\sigma^2$.
	Let  $X \sim \mathcal{N}(0,\sigma^2)$.
	By \cref{lem: gaussian conc}, we have
	$
	\mathbb{P}\left(D_i < t\right) \leq \mathbb{P}\left( \norm{X} < t \right)
	$ for all $t$.
	The right hand side is upper bounded by $2t/\sqrt{2 \pi} \sigma$, which is the product of maximal Gaussian density and length of interval.
	Apply union bound to get
	\[
	\mathbb{P}\left( \bigcap_{i=1}^{d+1} \{D_i \geq  t\}\right) \geq 1- \frac{2t(d+1)}{\sqrt{2 \pi }\sigma}.
	\]
	Let $C_i$ be the distance between the center of mass of $\conv(A)$ and $\aff(A_j:  j\neq i)$.
	Note that $C_i = D_i/(d+1)$.
	Then
	\[
	\mathbb{P}\left( \bigcap_{i=1}^{d+1} \{C_i \geq \frac{t}{d+1}\} \right) \geq 1- \frac{2t(d+1)}{\sqrt{2 \pi } \sigma}.
	\]
	The above expression states that with some probability the ball centered at the center of mass and of radius $t/(d+1)$ lies inside $\conv(A)$.
	Setting $t = \frac{\delta \sigma \sqrt{\pi}}{\sqrt{2}(d+1)}$ and using the fact that the width of the simplex is at least the diameter of the inscribed ball, we get 
	\[
	\mathbb{P}\left(\width \bigl (\conv(A) \bigr) \geq \sqrt{2 \pi}\sigma \delta  (d+1)^{-2} \right) \geq 1- \delta.
	\]
	The claim follows immediately from \cref{lem:minwidth} and \cref{thm:summary}.
\end{proof}

\subsubsection{Smoothed diameter}

\begin{lemma}\label{lem:simplex diameter}
	Let $A= \{A_1, \dotsc, A_{d+1}\}$ be a set of independent Gaussian random vectors with means $\mu_i$, $\norm{\mu_i} \leq 1, i \in [d+1]$, and covariance matrix $\sigma^2 \mathbf{I}_d$.
	Then for $\delta > 0$,
	\[
		\mathbb{P}\Bigg( \diam(A) \leq 2\Big( \sigma \sqrt{ 2d + 3 \ln \Bigl( \frac{d+1}{\delta} \Bigr)} + 1 \Big) \Bigg) \ge 1 - \delta.
	\]
\end{lemma}
\begin{proof}
	Let $A_i = \mu_i + X_i$, where $X_i \sim \mathcal{N}(0, \sigma^2 \mathbf{I}_d)$.
	Let $t>0$.
	Triangle inequality gives that
	\(
	\mathbb{P}(\|A_i\| > t+1) = \mathbb{P}(\|X_i + \mu_i \| > t+1) \le \mathbb{P}(\|X_i\|>t).
	\)
	Apply \cref{lem:laurent} with $\alpha = (\sigma^2,\dotsc,\sigma^2)$, we have
	\[
		\mathbb{P}\left( \|A_i\| > \sigma \sqrt{ d + 2\sqrt{dt} + 2t} + 1 \right) \leq \mathbb{P} \left( \norm{X_i} \geq \sigma \sqrt{ d + 2\sqrt{dt} + 2t} \right) \leq e^{-t},
	\]
	which shows that every $A_i$ is contained in a ball of radius $\sigma \sqrt{ d + 2\sqrt{dt} + 2t} + 1 \leq \sigma \sqrt{ 2d + 3t} + 1 $ with high probability.
	With union bound, we see the diameter of the ball is an upper bound of the diameter of convex hull of $A$:
	\[
	\mathbb{P}\Big( \diam \bigl(\conv(A) \bigr) \leq 2 \big( \sigma \sqrt{ 2d + 3t} + 1 \big) \Big) \ge 1 - (d+1) e^{-t}.
	\]
	The claim then follows by setting $t=\ln \bigl((d+1)/\delta\bigr)$.
\end{proof}

Next we restate and prove our main theorem for this section: 
\thmsimplexsmoothed*
\begin{proof}
	We proved in \cref{lem:simplex} and \cref{lem:simplex diameter} that,
	\[
		\mathbb{P} \left( \pwidth \bigl(\conv (A) \bigr) \geq  \sqrt{2 \pi}\sigma \delta  (d+1)^{-2} \right) \geq 1- \delta .
	\]
	and
	\[
		\mathbb{P}\Bigg( \diam(A) \leq 2\Big( \sigma \sqrt{ 2d + 3 \ln \Bigl( \frac{d+1}{\delta} \Bigr)} + 1 \Big) \Bigg) \ge 1 - \delta.
	\]
	Thus with probability at least $1-2\delta$, we have
	\begin{align*}
	\frac{\pwidth(A)}{\diam(A)}
	& \geq  \frac{\sqrt{2 \pi} \sigma \delta  (d+1)^{-2}}{  2\Big( \sigma \sqrt{ 2d + 3 \ln \left( \frac{d+1}{\delta} \right)} + 1 \Big)  } \\
	& =   \frac{ \delta \sqrt{\hfrac{\pi}{2}} }{ (d+1)^2  \Bigl( \sqrt{ 2d + 3 \ln \left( \frac{d+1}{\delta} \right)} + \frac{1}{\sigma} \Bigr)  } \\
	& \geq 1/ \rho(d, 1/\sigma, 1/\delta)
	\end{align*}
	where $\rho$ is a polynomial function of $d, \hfrac{1}{\sigma}, \hfrac{1}{\delta}$.
\end{proof}


Going back to \eqref{eqn:conv rate}, let $h_t = f(u_t) - f^* $. We have
\begin{align*}
h_t \leq \left(1-\frac{\mu}{4L}\left(\frac{\pwidth(A)}{\diam(A)}\right)^2\right)^t h_0.
\end{align*}
Based on our smoothed analysis on the measure of conditioning in \cref{thm:simplex smoothed}, with probability at least $1-2\delta$
\[
h_t \leq \left(1 - \frac{\mu}{4L \rho^2} \right)^t h_0 \leq e^{-\frac{\mu t}{4L \rho^2}} h_0.
\]
Hence one needs at most $\frac{ 4L \rho^2  \ln(\frac{1}{\eps})}{ \mu }$ good iterations to get a solution whose value is within distance $\eps(f_0 - f^*)$ of $f^*$.
Let $T$ denote the number of good iterations, we have (using the notation from \cref{def:psc})
\[
\max_{\substack{A \subseteq B(0,1) \subseteq \RR^d \\ \card{A}=d+1}} \prob_g \left( T(A+g) \geq \frac{ 4L \rho(d, \frac{1}{\sigma},\frac{1}{\delta})^2  \ln( \frac{1}{\eps})}{ \mu } \right) \leq 2\delta.
\]

\section{Conditioning of random matrices}\label{sec:matrixconditioning}

In this section we prove that the smallest singular value of some square submatrix of a $d$-by-$n$ Gaussian random matrix is exponentially small with probability exponentially close to 1 when $n/d \geq c > 1$.
From \cref{lem:RudelsonVershynin}, we know that the smallest singular value of a square matrix is comparable to the minimum distance between one column vector and the span of the other column vectors (one-off-distance).
If we consider exponentially narrow bands around each span of $d-1$ column vectors of a rectangular matrix, the matrix will have exponentially small minimum singular value if some other column vector falls in one of those bands.
We lower bound the Gaussian measure of the union of bands by a constant using the first two terms of the inclusion-exclusion principle (Bonferroni inequality).
See \cref{sec:introduction} for a high level overview of the proof.

We start by giving an upper bound of the intersection of two bands in Gaussian measure, which appears in the second term of the inclusion-exclusion principle. The following lemma shows that the Gaussian measure of the intersection depends on the width of bands and the angle between two bands.


\begin{lemma}\label{lem:parallelogram}
	Let $u,v \in \RR^d$ be unit length vectors, let $\eps > 0$, and let $c_S, c_T \in \RR$.
	Let
	\begin{align*}
		\mathcal{B}_S = \{x \in \RR^d \suchthat c_S \leq \inner{x}{u} \leq c_S +\eps \}, \\
		\mathcal{B}_T = \{x \in \RR^d \suchthat c_T \leq \inner{x}{v} \leq c_T +\eps \}.
	\end{align*}
	Then
	\[
		\mathcal{G}(\mathcal{B}_S \cap \mathcal{B}_T) \leq \frac{\eps^2}{\sqrt{2\pi ( 1-(\inner{u}{v})^2 )}}.
	\]
\end{lemma}
\begin{proof}
	If $u$ and $v$ are parallel then the claim holds.
	If they are not parallel, then by the structure of the Gaussian measure $\mathcal{G}$ this is a two-dimensional problem in the plane spanned by $u,v$.
	Identify this plane with $\RR^2$.
	$\mathcal{G}(\mathcal{B}_S \cap \mathcal{B}_T)$ is at most the maximum density $1/\sqrt{2 \pi}$ multiplied by the area of the parallelogram $P' := \{x \in \RR^2 \suchthat c_S \leq \inner{x}{u} \leq c_S +\eps, c_T \leq \inner{x}{u} \leq c_T +\eps  \}$.
	One can see that $P'$ has the same area as $P := \{x \in \RR^2 \suchthat \abs{\inner{x}{u}} \leq \eps/2, \abs{\inner{x}{v}} \leq \eps/2 \}$.
	Defining $A$ to be the matrix with rows $u,v$, we have $P = \{ x \suchthat \norm{Ax}_\infty \leq \eps/2\}$.
	This implies $\area(P) =  \eps^2 \abs{\det A^{-1}} =  \eps^2/\abs{\det A} = \eps^2/\sqrt{\det A A^T} = \eps^2/\sqrt{1-(\inner{u}{v})^2}$.
	The claim follows.
\end{proof}

We now switch our focus to the random regime.
The following lemma gives a probabilistic upper bound of the intersection of two bands around the spans of two (possibly not disjoint) subsets of random vectors  in high dimensional space.
The bound is good when not too many points are shared by the subsets (so that the behavior is not very different from two independent bands).
\begin{lemma}\label{lem:twobands}
	Let $d \geq 1$.
	Let $0 \leq k \leq d-1$.
	Let $A_1,\dotsc,A_k$, $S_1, \dotsc, S_{d-k-1}$, $T_1, \dotsc, T_{d-k-1}$ be $d$-dimensional iid.\ standard Gaussian random vectors.
	Let 
	\begin{align*}
	\mathcal{B}_S &= (\linspan \{A_1, \dotsc, A_k, S_1, \dotsc, S_{d-k-1} \})_{\eps/2}, \\
	\mathcal{B}_T &= (\linspan \{A_1, \dotsc, A_k, T_1, \dotsc, T_{d-k-1} \})_{\eps/2}.
	\end{align*}
	Then for any $t \geq 1$,
	\[
	\mathbb{P} \left(\mathcal{G}(\mathcal{B}_S \cap \mathcal{B}_T) \geq \frac{\eps^2 t}{\sqrt{2\pi}} \right) \leq  \frac{1}{t^{d-k-2}}.
	\]
\end{lemma}
\begin{proof}
	If $d \leq 2$ or $k \geq d-2$, then the claim is immediate.
	Otherwise, $0 \leq k \leq d-3$ and we argue in the following way:
	By the structure of the Gaussian measure $\mathcal{G}$ this is a $(d-k)$-dimensional problem in $\{A_1, \dotsc, A_k\}^\perp$.
	More precisely, let $U, V$ be two $(d-k)$-dimensional iid.\ uniformly random unit-length vectors and define $\mathcal{B}_S' = \{x \in \RR^{d-k} \suchthat 
	\abs{\inner{x}{U}} \leq \eps/2 \}$ and $\mathcal{B}_T' = \{x \in \RR^{d-k} \suchthat \abs{\inner{x}{V}} \leq \eps/2 \}$.
	Then $\mathcal{G}(\mathcal{B}_S \cap \mathcal{B}_T)$ has the same distribution as $\mathcal{G}(\mathcal{B}_S' \cap \mathcal{B}_T')$. 
	From \cref{lem:parallelogram} we have $\mathcal{G}(\mathcal{B}_S' \cap \mathcal{B}_T') \leq \frac{\eps^2}{\sqrt{2\pi (1-(\inner{U}{V})^2)}}$.
	
	
	Using the rotational symmetry of the distribution of $U$ and $V$ and then \cref{lem:archimedes} we get  
	\begin{align*}
	\mathbb{P}\bigl(\sqrt{(1-(\inner{U}{V})^2)} \leq 1/t\bigr) 
	&= \mathbb{P} \left( \sqrt{U_1^2 + \dotsb +  U_{d-k-1}^2} \leq 1/t \right) \\
	&\leq \mathbb{P} \left( \sqrt{U_1^2 + \dotsb +  U_{d-k-2}^2} \leq 1/t \right) \\
	&= 1/t^{d-k-2}.
	\end{align*}
	The claim follows.
\end{proof}

The main technical content of our singular value bound is the following lower bound on the Gaussian volume of the union of bands around any $d-1$ columns of a $d$-by-$n$ Gaussian random matrix. 
We also include an upper bound on the volume.
\begin{lemma}\label{lem:union of bands(non-asym)}
	Let $ \eps \geq 0$, $d\geq 2$.
	For \{$A_1, \dotsc, A_n\} \subseteq \RR^d$, define
	\[
	V = \mathcal{G} \left( \Bigl( \bigcup_{S \subseteq [n], \card{S}=d-1} \linspan A_S \Bigr)_\eps \right).
	\]
	\begin{enumerate}
	\item $V \leq \frac{2 \eps}{\sqrt{2\pi}}  \binom{n}{d-1} $.
	
	\item Suppose $A_1, \dotsc, A_n$ are $d$-dimensional iid.\ standard Gaussian random vectors with $\frac{n}{d-1} \geq c_0 > 1$.
	Then there exist constants $c_2, c_4 > 1$ (that depend only on $c_0$) such that when $\eps \leq 1/(c_4 c_2^{d-1})$ and with probability at least $1-c_4e^{-d}$ we have
	$
		V \geq \frac{c_2^{d-1}}{\sqrt{2 \pi}} \eps
	$.
	
	\end{enumerate}
\end{lemma}
\details{When $c_0$ is close to 1, small $n=d$ can satisfy $n/(d-1) \geq c_0$ but in this case $c_3$ is a large constant so that the probability $1-c_3e^{-d}$ is positive only for large $d$.}

\begin{proof}[Proof of part 1]
		The upper bound follows from the union bound and the fact that the 1-dimensional Gaussian density is upper bounded by $1/\sqrt{2\pi}$.
\end{proof}
\begin{proof}[Proof of part 2]
	Let $\mathcal{S} = \{S \subseteq [n], \card{S}=d-1 \}$. 
	Use \cref{lem:gilbert general} to get the bound $A\bigl(n,c_1(d-1),d-1\bigr) \geq c_2^{d-1}$. 
	We get a subfamily $\mathcal{T} \subseteq \mathcal{S}$ such that for all $S, T \in \mathcal{T}$ with $S \neq T$ we have $\card{S \cap T} \leq (1-\frac{c_1}{2}) (d-1)$ and $\card{\mathcal{T}} = c_2^{d-1}$ for some constants $0 < c_1 < 1$, $c_2 > 1$ (that depend only on $c_0$), and any $d \geq 2$.
	Let $N = \card{\mathcal{T}}$.
	
	Let $\mathcal{B}_S = (\linspan A_S)_\eps$.
	Use the first two terms of the inclusion-exclusion principle (Bonferroni inequality) and use \cref{lem:twobands} in a union bound applied to all pairs of sets in $\mathcal{T}$  to get  $\mathcal{G}(\mathcal{B}_S \cap \mathcal{B}_T) \leq \frac{4 \eps^2 t}{\sqrt{2\pi}}$ for all $S,T \in \mathcal{T}, S \neq T$. 
	We get a bound on $V$ that holds with probability at least 
	$1-\frac{\binom{N}{2}}{t^{d-1 - \left(1-\hfrac{c_1}{2}\right) (d-1)-2}} 
	= 1-\frac{\binom{N}{2}}{t^{\hfrac{c_1(d-1)}{2} -1}} 
	\geq 1- \frac{N^2}{t^{\hfrac{c_1(d-1)}{2} - 1}} 
	= 1- t \left(\frac{ c_2^{2}}{t^{c_1/2}}\right)^{d-1} 
	= 1 - c_3 e^{-d}
	$
	(choosing a constant $t > 1$ that depends on $c_1(c_0)$ and $c_2(c_0)$ such that $c_2^2/t^{c_1/2} = 1/e$ and then setting 
	$c_3 = t/e$, 
	which ultimately depends only on $c_0$).
	The bound on $V$ is
	\begin{align*}
	V 
	&\geq \mathcal{G} \left(\Bigl(\bigcup_{S \in \mathcal{T}} \linspan A_S\Bigr)_\eps \right) \\
	&\geq \sum_{S \in \mathcal{T}} \mathcal{G}(\mathcal{B}_S) - \frac{1}{2} \sum_{S,T \in \mathcal{T}, S \neq T} \mathcal{G}(\mathcal{B}_S \cap \mathcal{B}_T) \\
	&\geq \frac{2N\eps }{\sqrt{2 \pi}} e^{-\eps^2/2}- \binom{N}{2} \frac{4\eps^2 t}{\sqrt{2\pi}} \\
	&\geq \frac{ 2 N \eps}{\sqrt{2\pi}} (e^{-\eps^2/2} - 2 t N \eps) \\
	&\geq \frac{ 2 N \eps}{\sqrt{2\pi}} (1-\eps^2/2 - 2 t N \eps) \\
	&\geq \frac{N \eps}{\sqrt{2\pi}} \qquad \text{(for $\eps \leq 1/(8tN)$).}
	\end{align*}
In other words, $V \geq \frac{c_2^{d-1} \eps}{\sqrt{2\pi}}$ for $\eps \leq 1/(8e c_3 c_2^{d-1})$.	
We finish our proof by taking $c_4 = 8ec_3$.
\end{proof}

We are ready now to restate and prove the main results of the section.

\thmsigmamgeneral*
\begin{proof}
	Pick $c_1 \in (1,c_0)$.
	Let $m = \floor{c_1 d}$.
	Note that $m \geq c_1 d - 1 \geq c_1 d - c_1 \geq c_1(d-1)$, so that we can apply \cref{lem:union of bands(non-asym)} to columns $A_1, \dotsc, A_m$ with $\eps = \hfrac{1}{c_4c_2^{d-1}}$.
	Then we get $V \geq \frac{1}{\sqrt{2\pi} c_4 }$ with probability greater than $1-c_4e^{-d}$.
	This implies that with probability greater than 
	\begin{align*}
	(1-c_4e^{-d}) \bigl( 1-(1- \frac{1}{\sqrt{2\pi} c_4 } )^{n-m} \bigr) 
	& \geq (1-c_4 e^{-d}) \bigl(1- (1- \frac{1}{\sqrt{2\pi} c_4 } )^{(c_0-c_1)d} \bigr)\\
	& \geq 1-c_4 e^{-d} - (1-\frac{1}{\sqrt{2\pi} c_4 }  )^{(c_0-c_1)d}\\
	& \geq 1 - 2c_4 c_6^{d}
	\end{align*}
	where $c_6 = \max \{1/e, (1-\frac{1}{\sqrt{2\pi} c_4 } )^{(c_0-c_1)}\}$, 
	at least one of $A_{m+1},\dotsc, A_n$, say $A_*$, falls in $V$, that is, falls within distance $\eps = 1/c_4c_2^{d-1}$ of $\linspan (A_S)$ for some $S \subseteq [m], \card{S}=d-1$. 
	\Cref{lem:RudelsonVershynin} gives $\sigma_d(A_S,A_*) \leq  1/c_4c_2^{d-1}$.
\end{proof}

\thmsigmamgeneraltwo*
\begin{proof}
	Apply \cref{lem:union of bands(non-asym)} to columns $A_1, \dotsc, A_{n-1}$ to get
	\[
	V \leq  \frac{2\eps}{\sqrt{2\pi}} \binom{n}{d-1} \leq \frac{2\eps}{\sqrt{2\pi}} \left(\frac{en}{d-1} \right)^{d-1} \leq \frac{2\eps}{\sqrt{2\pi}} (eC_0)^{d-1}.
	\]
	By picking $\eps = 1/C_1^{d-1}$ where $C_1 > eC_0$, there exists a constant $\frac{eC_0}{C_1} < C_2 < 1$ such that $V \leq C_2^{d-1}$.
	This implies that, with probability at most $C_2^{d-1}$, column $A_n$ is within distance $ 1/C_1^{d-1}$ of $\linspan A_S$ for some $S \subseteq [n-1], \card{S}=d-1$.
	A similar claim holds for columns $A_1, \dotsc, A_{n-1}$ as well.
	Applying the union bound, we get that no $A_i$ falls within distance $1/C_1^{d-1}$ of $\linspan A_S$ for any $S \subseteq [n-1], \card{S}=d-1$ with probability at least $1 - nC_2^{d-1}$.
	\Cref{lem:RudelsonVershynin} gives $\sigma_d(A_S,A_n) \geq 1/C_1^{d-1}$ with probability at least $1 - nC_2^{d-1}$.
\end{proof}

\section{On the stability of tensor decomposition}\label{sec:tensor}

Kruskal \cite{MR444690} showed a sufficient condition under which the component vectors $a_i, b_i, c_i$, $i=1,\dotsc,n$ of an order-3 tensor $T=\sum_{i=1}^n a_i \otimes b_i \otimes c_i$ are uniquely determined by the tensor (up to inherent ambiguities).
The condition depends on a parameter now known as the Kruskal rank of a matrix: For a $d$-by-$n$ matrix $A$, the Kruskal rank of $A$, denoted $\krank(A)$, is the maximum $r \in [n]$ such that any $r$ columns of $A$ are linearly independent.
The condition is $\krank(A) + \krank(B) + \krank(C) \geq 2n + 2$, where $A$, $B$, $C$ are the matrices with columns $(a_i)$, $(b_i)$, $(c_i)$, respectively.
For concreteness, it is helpful to consider the symmetric case $A=B=C \in \RR^{d \times n}$.
Kruskal's condition becomes $3 \krank(A) \geq 2n+2$.
Informally, for a \emph{generic} matrix $A$ we have $\krank(A)=d$ and so Kruskal's result guarantees uniqueness for generic $A$ when $n\leq 3d/2-1$.

Bhaskara, Charikar and Vijayaraghavan~\cite[Theorem 5]{bhaskara2014uniqueness} extended Kruskal's uniqueness to a result that guarantees \emph{robust decomposition}.
That is, when the observed tensor is a small perturbation of the original tensor, the components of the perturbed tensor are uniquely determined and close to the components of the original tensor.
Their condition for robust unique decomposition is a refinement of Kruskal's condition:
Let $\tau > 0$.
The \emph{robust Kruskal rank (with threshold $\tau$)} of $A$, denoted $\rkrank{\tau}{A}$, is the maximum $k \in [n]$ such that for any subset $S \subseteq [n]$ of size $k$ we have $\sigma_k(A_S) \geq 1/\tau$ ($\sigma_k$ denotes the $k$th largest singular value).
The condition is $\rkrank{\tau}{A} + \rkrank{\tau}{B} + \rkrank{\tau}{C} \geq 2n + 2$ and the error in the recovered components depends polynomially on $\tau$.

In this context, \cref{thm:sigmam general} can be stated in the following equivalent way:
\begin{theorem}
	Let $A$ be an $d$-by-$n$ random matrix with iid.\ standard Gaussian entries with $d \geq 2$ and $n/d \geq c_0> 1$.
	Then, there exist constants $c_4, c_5 > 1$, $0 < c_6 < 1$ (that depend only on $c_0$) such that with probability at least $1-2c_4 c_6^d$, 
	\[
		\rkrank{\tau}{A} = d \Rightarrow \tau \geq  c_4c_5^{d-1}.
	\]
\end{theorem}
This has the following implication for Bhaskara, Charikar and Vijayaraghavan's result:
Even though Kruskal's result guarantees uniqueness for generic $A$ when $n = 3d/2-1$ (say, with probability 1 for a random Gaussian matrix we have $\krank(A)=d$), Bhaskara, Charikar and Vijayaraghavan's robust uniqueness can give a polynomial bound on the reconstruction error on no more than an exponentially small fraction of matrices $A$ when the fraction is measured by the Gaussian measure.
This rarity of sufficiently well-conditioned matrices $A$ is somewhat surprising.
\details{not clear what happens for $n=cd$ and aiming for $\rkrank{\tau}{A} = c'd$, $c>1$ and $c'<1$}

\section{On the complexity of the simplex method and the diameter of polytopes}\label{sec:deltadistance}
In \cite{brunsch2013finding}, Brunsch and R\"oglin introduced the following property of a matrix:
\begin{definition}[{$\delta$-distance property, \cite{DBLP:conf/stacs/BrunschGR15}}]
Let $A = (a_1, \dotsc, a_m)^\top$ be an $m$-by-$n$ matrix with unit rows.
We say that $A$ satisfies the $\delta$-distance property if: for any $I \subseteq [m]$ and any $j \in [m]$ whenever $a_j \notin \linspan \{a_i \suchthat i \in I\}$ we have $d(a_j, \linspan \{a_i \suchthat i \in I\}) \geq \delta$.
\end{definition}
This property has been used in several papers \cite{brunsch2013finding,%
DBLP:conf/stacs/BrunschGR15,%
DBLP:journals/dcg/DadushH16,%
DBLP:journals/mp/EisenbrandV17%
} to study polytopes of the form $\{x \in \RR^n \suchthat Ax \leq b\}$ to provide upper bounds of the form $\poly(m,n,1/\delta)$ on their diameter and the number of pivot steps of the simplex method.
Our \cref{thm:sigmam general} combined with \cref{lem:RudelsonVershynin} and concentration of the length of a Gaussian random vector implies that, for $m/n \geq c' > 1$, matrices $A$ with the $\delta$-distance property for $\delta \geq c^n$, $0<c<1$, are ``rare'': they are exponentially unlikely when the rows are iid.\ random unit vectors.
As in \cref{sec:tensor}, this rarity of well-conditioned matrices $A$ is somewhat surprising.

\section{On the smoothed analysis of polytope conditioning}\label{sec:polytopeconditioning}
	In this section we prove that the vertex-facet distance of the convex hull of a linear number of $d$-dimensional iid.\ Gaussian points can be exponentially small with probability at least some constant.
	The argument is a more elaborate version of the argument for the minimum singular value in \cref{sec:matrixconditioning} and works in the following way.
	\Cref{fig:vf small} shows a polytope, the convex hull of a partial sequence of random points, and $\eps$-innner bands at all facets.
	If a new point falls into the blue region, then the new polytope, which is the convex hull of the old polytope plus the new point, will have vertex-facet distance no larger than $\eps$:
	the new point is a vertex and its distance to the affine hull of the facet associated to the band where the point lies in is less than $\eps$.
	
	To get a lower bound on the Gaussian measure of the blue region (\cref{lem:bands2 lim} \lnote{check}), we add the measures of the bands and then we subtract the measures of pairwise intersections of bands and $\eps$-inner neighbourhood (grey region). 
	\Cref{lem:twobandsaff} gives a bound on the measure of a pairwise intersection.
	Its proof is divided into two cases: \cref{lem:twobandsaff no k} for the case where the two facets do not share vertices and \cref{lem:twobandsaff2} for the case where they do share vertices. This argument is a refinement of the proof of \cref{lem:twobands}.
\begin{figure}[ht]
	\centering
	\includegraphics[width=0.4\textwidth,trim=60pt 70pt 80pt 70pt, clip]{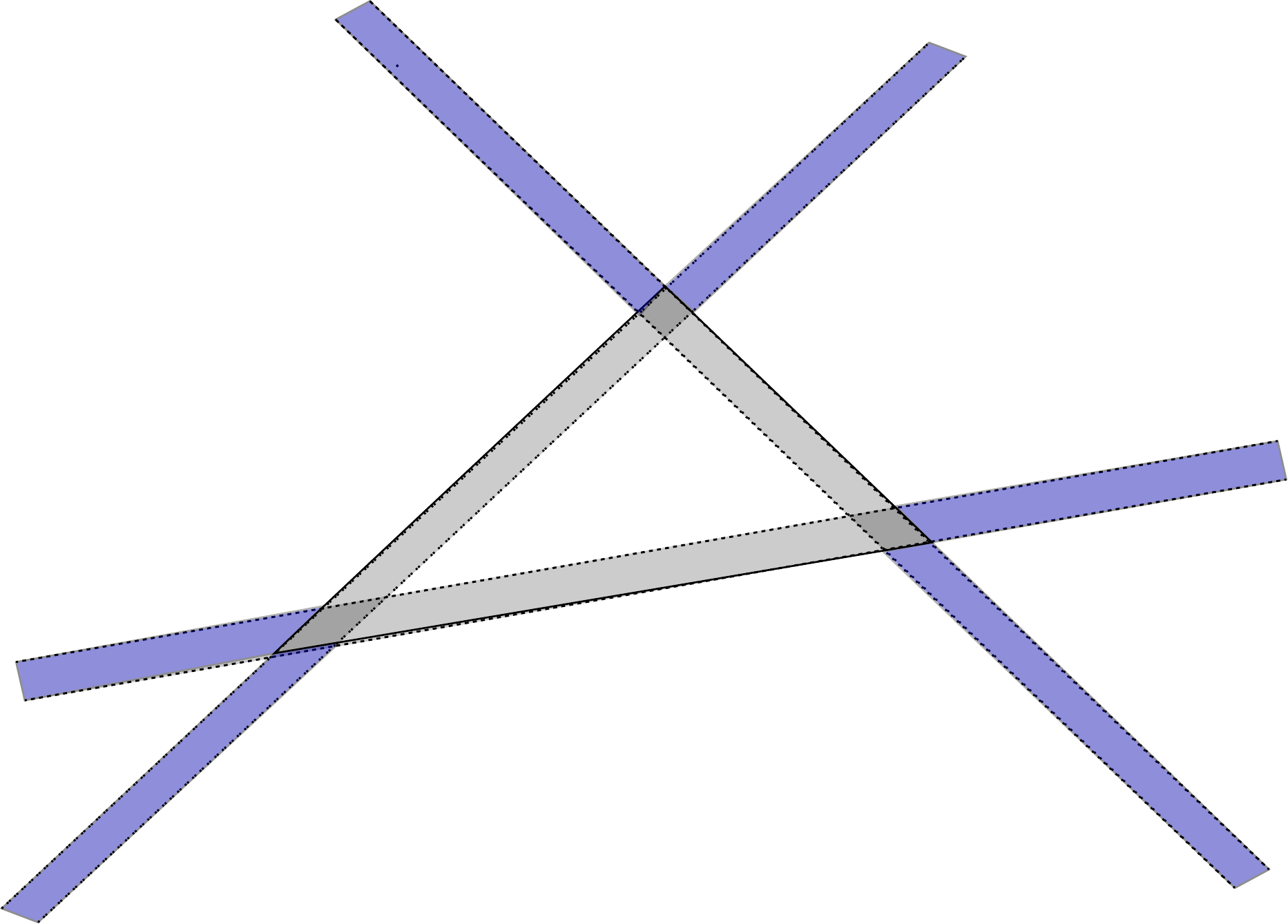}
	\caption{A polytope (triangle) and the region (blue) where a new point would create a small vertex-facet distance.}
	\label{fig:vf small}
\end{figure}


\cnote{remark: \cref{lem:twobandsaff} can be easily generalized to covariance matrix $\sigma^2 \mathbf{I}_d$.}

\ifstoc
\else
\begin{lemma}\label{lem:twobandsaff no k}
	Let $S_1, \dots, S_{d}$, $T_1, \dots, T_{d}$ be iid.\ standard Gaussian random vectors in $\RR^d$.
	Let 
	\begin{align*}
	\mathcal{B}_S &= \bigl(\aff \{ S_1, \dotsc, S_{d} \}\bigr)_{\eps /2}, \\
	\mathcal{B}_T &= \bigl(\aff \{ T_1, \dotsc, T_{d} \}\bigr)_{\eps /2}.
	\end{align*}
	Then for $t \geq 1$
	\[
	\mathbb{P} \left(\mathcal{G}(\mathcal{B}_S \cap \mathcal{B}_T) \geq \frac{\eps^2 t}{\sqrt{2\pi}} \right) \leq  \frac{1}{t^{d-2}}.
	\]
\end{lemma}
\begin{proof}
	By the rotational invariance of the Gaussian distribution, unit normal vectors $U, V$ to $\mathcal{B}_S, \mathcal{B}_T$ are independent and are uniformly distributed on $\mathcal{S}^{d-1}$.
	Define 
	\begin{align*}
	\mathcal{B}_S' = \{x \in \RR^{d} \suchthat \abs{\inner{x}{U}} \leq \eps/2 \},\\
	\mathcal{B}_T' = \{x \in \RR^{d} \suchthat \abs{\inner{x}{V}} \leq \eps/2 \}.
	\end{align*}
	By a standard argument (say, using logconcavity) we have $\mathbb{P} \bigl( \mathcal{G}(\mathcal{B}_S \cap \mathcal{B}_T) \geq t  \bigr) \leq \mathbb{P} \bigl( \mathcal{G}(\mathcal{B}_S' \cap \mathcal{B}_T') \geq t  \bigr) $. 
	Then, by the argument in the proof of \cref{lem:twobands} we get that for any $t \geq 1$,
	\(
	\mathbb{P} \left(\mathcal{G}(\mathcal{B}_S' \cap \mathcal{B}_T') \geq \frac{\eps^2 t}{\sqrt{2\pi}} \right) \leq  \frac{1}{t^{d-2}}.
	\)
	The claim follows.
\end{proof}

\begin{lemma}\label{lem:twobandsaff2}
Let $A_1,\dots,A_k$, $S_1, \dots, S_{d-k}$, $T_1, \dots, T_{d-k}$ be iid.\ standard Gaussian random vectors in $\RR^d$, and $1 \leq k \leq d$.
Let 
\begin{align*}
\mathcal{B}_S &= \bigl( \aff \{A_1, \dotsc, A_k, S_1, \dotsc, S_{d-k} \} \bigr)_{\eps /2}, \\
\mathcal{B}_T &= \bigl( \aff \{A_1, \dotsc, A_k, T_1, \dotsc, T_{d-k} \} \bigr)_{\eps /2}.
\end{align*}
Then for $ 0 < 2\alpha \leq \beta < \pi/ 2 $,
\begin{equation}\label{eqn:bandintersection}
\mathbb{P} \left(\mathcal{G}(\mathcal{B}_S \cap \mathcal{B}_T) \geq \frac{\epsilon^2 }{\sqrt{2\pi} \sin \alpha} \right) \leq   (\sin \beta)^{d-k-1} +  2 \left(\frac{\sin \alpha}{\sin(\beta - \alpha)}\right) ^{d-k-2}.
\end{equation}
In particular, for $t > 2\pi$ we have
\[
\mathbb{P} \left(\mathcal{G}(\mathcal{B}_S \cap \mathcal{B}_T) \geq \frac{\epsilon^2  t}{\sqrt{2 \pi}} \right) 
\leq 3 \left(\frac{ \pi ^{\hfrac{3}{2}} }{ \sqrt{2t} } \right)^{d-k-2}.
\]

\end{lemma}

\begin{proof}
	If $d-k \leq 2$, then the bound holds immediately.
	Otherwise, $d -k > 2$ and we argue in the following way.
	By the structure of the Gaussian measure, this reduces to a $(d-k+1)$-dimensional problem:
	Conditioning on $A_i = a_i$, $i = 1,\dotsc,k$, we project onto the orthogonal complement of the linear subspace parallel to $\aff \{a_1, \dotsc, a_k\}$.
	We will then prove the bound claimed in \eqref{eqn:bandintersection} conditioning on $A_1, \dotsc, A_k$, which implies the claimed bound by total probability.

	With a slight abuse of notation, we denote the projection of $\aff \{a_1, \dotsc, a_k\}$ as $a_1$ and the projections of $S_i, T_i$ as $S_i, T_i,\ i = 1,\dotsc,d-k$. 
	Using the fact that the Gaussian distribution is rotationally invariant, we may assume without loss of generality that  $a_1 = \mu e_1$ for some $\mu \geq 0$. 
	A normal vector to $\aff \{a_1, S_1, \dots, S_{d-k} \}$ is\footnote{In the formula for $U$, the determinant should be interpreted as a formal cofactor expansion along the first row; the entries in the first row are the canonical vectors and the expansion gives the coefficients of these vectors (as subdeterminants).}
	\begin{align}\label{eq:normal} 
		U = \det \begin{pmatrix}
		e_1 & e_2 & \cdots & e_{d-k+1}\\
		& & P_1^T & & \\
		& & \vdots & & \\
		& & P_{d-k}^T & &
		\end{pmatrix},
	\end{align}
	where $P_i := S_i-a_1$, $i \in \left[d-k\right]$.
	Define the matrix $P = \begin{pmatrix}	P_1 \cdots P_{d-k} \end{pmatrix}$.
	Let $V$ be a normal vector to $\aff \{a_1, T_1, \dots, T_{d-k} \}$, defined similarly.
	
	Set $\bigl(\begin{smallmatrix} H_i \\ P_i' \end{smallmatrix}\bigr) = P_i $, where $H_i \sim \mathcal{N}(\mu,1)$ and $P_i' \sim \mathcal{N}(0, \mathbf{I}_{d-k})$.
	Denote $H^T = \begin{pmatrix} H_1 \cdots H_{d-k} \end{pmatrix}$ as the first row of matrix $P$ and $P' = \begin{pmatrix}	P_1' \cdots  P_{d-k}' \end{pmatrix}$ as the rest. $H$ and $P'$ are independent.
	
	Note that $\norm{U}^2 = \det(P^TP)$ (follows from \eqref{eq:normal} and the Cauchy-Binet formula).
	Also, $U_1 = U\cdot e_1= \det(P^{'})$.
	We now compute the distribution of the first coordinate of unit normal vector $\hat{U}$ (using the \emph{matrix determinant lemma} to compute the determinant of a rank-1 update).
	\begin{align*}
		\hat{U}_1^2 &= \frac{\det(P'^TP')}{\det(P^TP)} \\
		&=\frac{\det(P'^TP')}{\det(P'^TP'+HH^T)}\\
		&=\frac{\det(P'^TP')}{(1+H^TP'^{-1}P'^{-T}H)\det(P'^TP')}\\
		&=\frac{1}{1+H^TP'^{-1}P'^{-T}H}.
	\end{align*}

	\begin{claim}\label{claim: first coord} 
		We have $H^TP'^{-1}P'^{-T}H \deq \frac{\sum_{i=1}^{d-k}Y_i^2}{Y_0^2}$,
		where $Y_0 \sim \mathcal{N}(0,1)$, $Y_i \sim \mathcal{N}(\mu,1)$, $i \in \left[d-k\right]$ and $Y_0, Y_1, \dotsc, Y_{d-k}$ are independent.
	\end{claim}
	\begin{claimproof}
		Random variables $P'$ and $H$ are independent.
		Moreover, $P'$ is a Gaussian matrix and therefore the distribution of $P'^{-1}$ is invariant under any orthogonal transformation applied to rows or columns. 
		Thus, it is enough to consider the case $H = \norm{H}e_1$. 
		Note that $\norm{H}^2 \deq \sum_{i=1}^{d-k}Y_i^2$, and 
		$
		e_1^TP'^{-1}P'^{-T}e_1 
		= \norm{\text{first row of $P'^{-1}$}}^2 
		\deq \frac{1}{Y_0^2}
		$.\details{last step: by definition of inverse, $\norm{\text{first row of $P'^{-1}$}} = 1/\text{distance of first column of $P'$ to span of the rest}$.} 
		The claim follows.
	\end{claimproof}

	Recall that $\hat{U}, \hat{V}$ are unit normal vectors to $\mathcal{B}_S, \mathcal{B}_T$, respectively.
	We aim to show that $\mathbb{P} \bigl(\hat{V} \in \mathcal{C}_\alpha(\hat{U}) \cup \mathcal{C}_\alpha(- \hat{U})\bigr)$, i.e. $\prob\bigl(\abs{\hat{U}\cdot \hat{V}} \geq \cos \alpha \bigr)$, is upper bounded by an expression of the form $c(\alpha)^d$ with $c(\alpha) \to 0$ as $\alpha \to 0$ (where $\mathcal{C}_\alpha(\hat U)$ denotes the spherical cap centered at $\hat{U}$ with angle $\alpha$).
	To see this, we divide the analysis into two cases, depending on whether the cap is close to $e_1$.
	The case analysis depends on a parameter $\beta$ that will need to satisfy the constraint $\beta \geq 2 \alpha$. 

	\paragraph{Case 1:}
	$\mathcal{C}_\alpha(\hat{U}) \subseteq \mathcal{C}_\beta(e_1) \cup \mathcal{C}_\beta(-e_1)$ (equivalently, $\abs{\hat{U}_1} \geq \cos (\beta - \alpha)$).

	In this case, the $\alpha$-cap around $\hat{U}$ is contained in a larger cap centered at $e_1$.
	\begin{align}
	\mathbb{P} \Bigl(\bigl\{ \hat{V} \in \mathcal{C}_\alpha(\hat{U}) \cup \mathcal{C}_\alpha(- \hat{U}) \bigr\}  \cap \bigl\{ \mathcal{C}_\alpha(\hat{U}) \subseteq \mathcal{C}_\beta(e_1) \cup \mathcal{C}_\beta(-e_1) \bigr\}\Bigr) 
	&\le \mathbb{P} \bigl(\hat{V} \in \mathcal{C}_\beta(e_1) \cup \mathcal{C}_\beta(-e_1)\bigr) \nonumber \\
	\text{(using $\beta \leq \pi/2$)} \qquad &= \mathbb{P} (\hat{V}_1^2 \ge \cos^2 \beta).\label{eq:v1hatquared}
	\end{align} 
	From \cref{claim: first coord} we get
	\begin{align*}
		\hat{V}_1^2 \overset{d}{=} \frac{Y_0^2}{Y_0^2 + \sum_{i=1}^{n} Y_i^2}.
	\end{align*}
	To upper bound \eqref{eq:v1hatquared}, we get from \cref{lem: monotonicity} that making $a_1=0$ (equivalently, $\mu=0$) only makes the rhs larger and we then bound the case $a_1 = 0$ explicitly.
	More precisely, let $W$ be a normal vector to $\linspan \{T_1, \dotsc, T_{d-k}\}$ defined similarly to $U$ and $V$:
	\begin{align*} 
		W = \det \begin{pmatrix}
		e_1 & e_2 & \cdots & e_{d-k+1}\\
		& & T_1^T & & \\
		& & \vdots & & \\
		& & T_{d-k}^T & &
		\end{pmatrix}.
	\end{align*}
	Note that $\hat W$ is a uniformly random unit vector.
	Following the same computation as for $V$, one can derive
	\[
	\hat{W}_1^2 \overset{d}{=} \frac{X_0^2}{X_0^2 + \sum_{i=1}^{n} X_i^2},
	\]
	where $X_0, X_i \sim \mathcal{N}(0,1)$, $i \in \left[d-k\right]$.
	Then, by \cref{lem: monotonicity},
	\( 
		\mathbb{P}(\hat{V}_1^2 \geq \cos^2 \beta) \leq	\mathbb{P}(\hat{W}_1^2 \geq \cos^2 \beta)
	\).
	Hence,
	\begin{align*}
		\mathbb{P} \Bigl(\bigl\{ \hat{V} \in \mathcal{C}_\alpha(\hat{U}) \cup \mathcal{C}_\alpha(- \hat{U}) \bigr\}  \cap \bigl\{ \mathcal{C}_\alpha(\hat{U}) \subseteq \mathcal{C}_\beta(e_1) \cup \mathcal{C}_\beta(-e_1) \bigr\}\Bigr) 
		&\leq \mathbb{P}(\hat{W}_1^2 \geq \cos^2 \beta)\\
		&= \mathbb{P} \left(\sum_{i=2}^{d-k+1} \hat{W}_i^2 \leq \sin^2 \beta \right)\\
		&\leq \mathbb{P} \left(\sum_{i=2}^{d-k} \hat{W}_i^2 \leq \sin^2 \beta \right)\\
		&\leq (\sin \beta)^{d-k-1}  \quad \text{(\cref{lem:archimedes})}. 
	\end{align*}

\paragraph{Case 2:}
$\mathcal{C}_\alpha(\hat{U})  \not \subseteq \mathcal{C}_\beta(e_1) \cup \mathcal{C}_\beta(-e_1)$. \details{equivalently $\abs{\hat{U}_1} < \cos (\beta - \alpha)$ .}
	
	If $\mathcal{C}_\alpha(\hat{U}) $ is not contained in $\mathcal{C}_\beta(e_1) \cup \mathcal{C}_\beta(-e_1)$, then $\hat{U}$ makes an angle at least $\beta - \alpha$ with $e_1$ and $-e_1$, that is
	\begin{align}
		\abs{\hat{U}_1} < \cos (\beta - \alpha) \label{eq:u1norm}.
	\end{align}
	Our goal here is to bound
	\begin{align}
		\mathbb{P} \Bigl(\bigl\{ \hat{V} \in \mathcal{C}_\alpha(\hat{U}) \cup \mathcal{C}_\alpha(- \hat{U}) \bigr\} & \cap    \bigl\{ \mathcal{C}_\alpha(\hat{U}) \not \subseteq \mathcal{C}_\beta(e_1) \cup \mathcal{C}_\beta(-e_1) \bigr\} \Bigr) \nonumber\\
		& \qquad =
			\prob \Bigl(\bigl\{ \hat{V} \in \mathcal{C}_\alpha(\hat{U}) \bigr\}  \cap \bigl\{ \mathcal{C}_\alpha(\hat{U}) \not \subseteq \mathcal{C}_\beta(e_1) \cup \mathcal{C}_\beta(-e_1) \bigr\}\Bigr) \nonumber\\
		&\qquad \quad + \mathbb{P} \Bigl(\bigl\{ \hat{V} \in  \mathcal{C}_\alpha(- \hat{U}) \bigr\}  \cap \bigl\{ \mathcal{C}_\alpha(\hat{U}) \not \subseteq \mathcal{C}_\beta(e_1) \cup \mathcal{C}_\beta(-e_1) \bigr\} \Bigr) \nonumber \\
	   & \qquad = 2 \cdot \mathbb{P} \Bigl(\bigl\{ \hat{V} \in \mathcal{C}_\alpha(\hat{U}) \bigr\}  \cap \bigl\{ \mathcal{C}_\alpha(\hat{U}) \not \subseteq \mathcal{C}_\beta(e_1) \cup \mathcal{C}_\beta(-e_1) \bigr\}\Bigr) \label{eq:case2hat}. 
	\end{align}
	
	Observe that the distribution of $\hat{U}$ and the distribution of $\hat{V}$ are invariant under rotations orthogonal to $e_1$.
	Thus, if we let $\hat{U}_{-1}, \hat{V}_{-1}$ be the projections of $\hat{U}, \hat{V}$ orthogonal to $e_1$ and $\widehat{U_{-1}}, \widehat{V_{-1}}$ be their normalizations, respectively, then $\widehat{U_{-1}}, \widehat{V_{-1}} \sim \operatorname{Unif}(\mathcal{S}^{d-k})$.
	
	This observation motivates us to use the corresponding probability of projections to bound \eqref{eq:case2hat}.
	We will show that under condition \eqref{eq:u1norm} of case 2, $\hat{V} \in \mathcal{C}_\alpha(\hat{U})$ implies that $\widehat{V_{-1}} \in \mathcal{C}_{f(\alpha)}(\widehat{U_{-1}})$, where $f(\alpha)$ is a bound (to be understood) on the angle that depends only on $\alpha$.
	As events,
	\begin{align}
	\{ \hat{V} \in \mathcal{C}_\alpha(\hat{U})\}
	&\subseteq \{ \hat{V}_{-1} \in \operatorname{Proj}_{e_1^{\perp}} \mathcal{C}_\alpha(\hat{U})  \}  \nonumber\\
	&\subseteq \{ \widehat{V_{-1}} \in \mathcal{C}_{f(\alpha)}(\widehat{U_{-1}})\}   \label{eq: case2bound}.
	\end{align}
	

	\begin{figure}[t]
	\centering
	\includegraphics[width = \columnwidth]{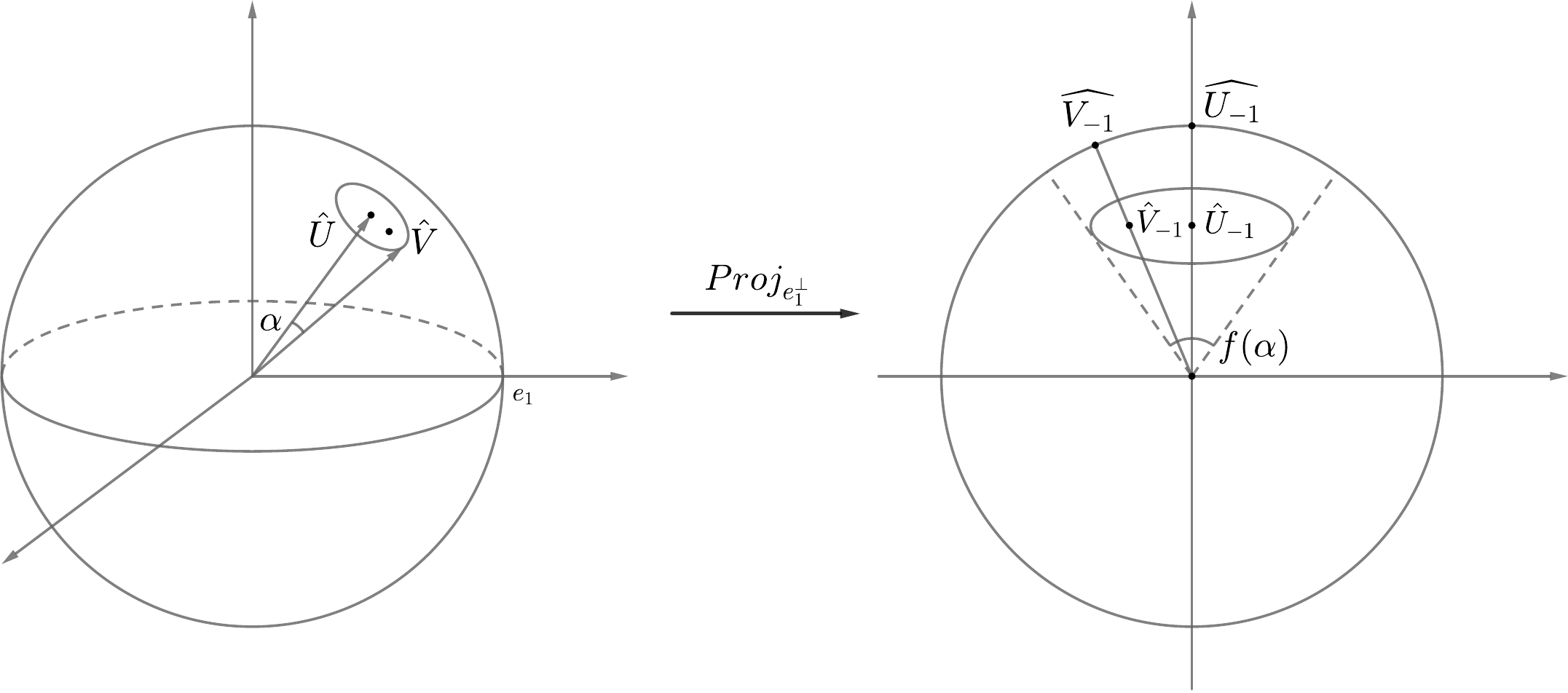}
	\caption{Case 2 of proof of \cref{lem:twobandsaff}}
	\label{fig:proj}
	\end{figure}
	
	Bounding $f(\alpha)$ is a three-dimensional problem since $\widehat{U_{-1}}, \widehat{V_{-1}}$ are in $\linspan\{e_1, \hat{U}, \hat{V}\}$.
	From now on, the analysis lives in the above three-dimensional space to get an upper bound on $f(\alpha)$.
	Let $\tilde{e}_2 = (\hat{U} - \hat{U}\cdot e_1)/ \norm{\hat{U} - \hat{U}\cdot e_1}$ (so that $\{ e_1,\tilde{e}_2\}$ is an orthonormal basis of $\linspan \{e_1, \hat U \}$).
	Let $\{ e_1,\tilde{e}_2,\tilde{e}_3 \}$ be an orthonormal basis of $\linspan\{e_1, \hat{U}, \hat{V}\}$,
	and let $\hat{U} = (\hat{U}_1,\hat{U}_2,0)$ be the coordinate tuple of $\hat{U}$ relative to $\{ e_1,\tilde{e}_2,\tilde{e}_3 \}$.
	Consider $x \in \mathcal{C}_{\alpha}(\hat{U})$ such that $x \cdot \hat{U} = \cos \gamma$. 
	Note that its coordinates $(x_1, x_2, x_3)$ in our chosen basis satisfy the following system of equations:
	\begin{align*} 
	x_1^2 + x_2^2 + x_3^2 &= 1 \\
	x_1\hat{U}_1 + x_2\hat{U}_2 &= \cos \gamma.  
	\end{align*}
	The projections of all such $x$ (for fixed $\gamma$) onto $\linspan\{\tilde{e}_2, \tilde{e}_3\}$ form the ellipse:
	\details{
	\begin{align*}
		x_1\hat{U}_1 + x_2\hat{U}_2 = \cos \gamma &\implies x_1 = \frac{\cos \gamma -x_2\hat{U}_2}{\hat{U}_1}\\
		x_1^2 + x_2^2 + x_3^2 = 1 &\implies (\frac{\cos \gamma -x_2\hat{U}_2}{\hat{U}_1})^2 + x_2^2 + x_3^2 = 1\\
		&\implies \cos ^2 \gamma + x_2^2 \hat{U}_2^2 - 2x_2\hat{U}_2\cos \gamma + x_2^2 \hat{U}_1^2 + x_3^2\hat{U}_1^2=\hat{U}_1^2\\
		&\implies (\hat{U}_1^2 + \hat{U}_2^2)x_2^2  - 2x_2\hat{U}_2\cos \gamma + \cos ^2 \gamma+x_3^2\hat{U}_1^2=\hat{U}_1^2\\
	\hat{U}_1^2 + \hat{U}_2^2 = 1	&\implies (x_2 - \hat{U}_2 \cos \gamma)^2 - \hat{U}_2^2 \cos^2 \gamma+ \cos ^2 \gamma+x_3^2\hat{U}_1^2=\hat{U}_1^2\\
		&\implies (x_2 - \hat{U}_2 \cos \gamma)^2 + \hat{U}_1^2 \cos ^2 \gamma +x_3^2\hat{U}_1^2=\hat{U}_1^2\\
		&\implies (x_2 - \hat{U}_2 \cos \gamma)^2  +x_3^2\hat{U}_1^2 = \hat{U}_1^2 \sin ^2 \gamma
	\end{align*}
	}
	\[ 
		(x_2 - \hat{U}_2 \cos \gamma)^2  +x_3^2\hat{U}_1^2 = \hat{U}_1^2 \sin ^2 \gamma.
	\]
	
	If $\hat{U}_1 = 0$, then $\hat{U}_2 = 1$, and the projection is the line segment inside unit circle at $x_2 = \cos \gamma$.
	The angle between $x_{-1}$ and $\widehat{U_{-1}}$ is upper bounded by $\gamma$.
	As $\gamma$ ranges from $0$ to $\alpha$, $\widehat{U_{-1}}$ and $\widehat{V_{-1}}$ form an angle at most $\alpha$ when $\hat{U}_1 = 0$.

	If  $\hat{U}_1 \neq 0$, the projection is an ellipse inside the unit circle.
	As shown in \cref{fig:proj}, angle between $x_{-1}$ and $\widehat{U_{-1}}$ can be upper bounded by angle formed by $\widehat{U_{-1}}$ and tangent line $x_2 = \frac{\sqrt{\cos^2 \gamma-\hat{U}_1^2}}{\sin \gamma} x_3$.
	Note that from \eqref{eq:u1norm} we know $\hat{U}_1^2 < \cos^2(\beta - \alpha) \leq \cos^2\alpha \leq \cos^2\gamma$ (here we use $\beta \geq 2\alpha$ explicitly), so the tangent line always exists.
	\details{		
	Suppose the tangent line is $x_2=kx_3$. To find $k$, we plug the line into the ellipse, then set the discriminant to 0.
	\begin{align*}
		&(x_2 - \hat{U}_2 \cos \gamma)^2  +x_3^2\hat{U}_1^2 = \hat{U}_1^2 \sin ^2 \gamma\\
		\iff & k^2x_3^2 - 2kx_3\hat{U}_2 cos \gamma + \hat{U}_2^2\cos^2 \gamma +x_3^2\hat{U}_1^2 = \hat{U}_1^2 \sin ^2 \gamma\\
		\iff &( k^2 + \hat{U}_1^2 )x_3^2 - 2kx_3\hat{U}_2 cos \gamma + (\hat{U}_2^2\cos^2  - \hat{U}_1^2 \sin ^2 \gamma) =0.
	\end{align*}
	Setting discriminant of the above function of $x_3$ equal to zero:
	\begin{align*}
		&\Delta = 4k^2\hat{U}_2^2 \cos^2 \gamma - 4(k^2 + \hat{U}_1^2)(\hat{U}_2^2\cos^2\gamma - \hat{U}_1^2\sin^2\gamma) = 0\\
		\implies& k^2\hat{U}_1^2\sin^2\gamma - \hat{U}_1^2\hat{U}_2^2\cos^2\gamma + \hat{U}_1^4\sin^2\gamma = 0\\
		\implies& \hat{U}_1 = 0 \text{ or } k^2 = \frac{\cos^2 \gamma - \hat{U}_1^2}{\sin^2 \gamma}.
	\end{align*}
	}
	
	Hence the angle between $x_{-1}$ and $\widehat{U_{-1}}$ is at most $\arctan \Big( \frac{\sin \gamma}{\sqrt{\cos^2\gamma-\hat{U}_1^2}}\Big).$
	Furthermore, since $\arctan \Big( \frac{\sin \gamma}{\sqrt{\cos^2\gamma-\hat{U}_1^2}}\Big)$ is increasing in $\gamma$,
	we can conclude that for any $\hat{V} \in \mathcal{C}_\alpha(\hat{U})$, its normalized projection orthogonal to $e_1$, $\widehat{V_{-1}}$, is contained in the spherical cap centered at $\widehat{U_{-1}}$ with polar angle at most $\arctan \Big( \frac{\sin \alpha}{\sqrt{\cos^2\alpha-\hat{U}_1^2}}\Big)$ when $\hat{U}_1 \neq 0$.\details{also $=\arctan \Big( \frac{\hat U_2^2}{\sin^2 \alpha}-1\Big)^{-1/2}$}
	
	Therefore with \eqref{eq:u1norm}, we can take
	\begin{align*}
		f(\alpha) = \max \{\arctan \Big( \frac{\sin \alpha}{\sqrt{\cos^2\alpha-\cos^2 (\beta - \alpha)}}\Big), \alpha \}.
	\end{align*}
	Combine with \eqref{eq:case2hat} and \eqref{eq: case2bound},
	\begin{multline*}
		\mathbb{P} (\{ \hat{V} \in \mathcal{C}_\alpha(\hat{U}) \cup \mathcal{C}_\alpha(- \hat{U}) \}  \cap \{ \mathcal{C}_\alpha(\hat{U}) \not \subseteq \mathcal{C}_\beta(e_1) \cup \mathcal{C}_\beta(-e_1) \}) \\
		\begin{aligned}
			&\leq 2 \cdot \mathbb{P} (\{ \widehat{V_{-1}} \in \mathcal{C}_{f(\alpha)}(\widehat{U_{-1}})\}  \cap \{ \mathcal{C}_\alpha(\hat{U}) \not \subseteq \mathcal{C}_\beta(e_1) \cup \mathcal{C}_\beta(-e_1) \})\\
			&\leq 2 \cdot \mathbb{P}\left(\abs{\widehat{U_{-1}} \cdot \widehat{V_{-1}}} \geq \cos (f(\alpha))\right) \\
			& =  2 \cdot \mathbb{P}\left( \sqrt{1 - (\widehat{U_{-1}} \cdot \widehat{V_{-1}})^2} \leq \sin f(\alpha) \right) \\
			&\leq 2 \left(\sin f(\alpha) \right)^{d-k-2} \qquad \text{(\cref{lem:archimedes})} \\
			&= 2\left(\max \{\frac{\sin \alpha}{\sqrt{\cos^2\alpha - \cos^2 (\beta - \alpha)  + \sin^2 \alpha}} , \sin \alpha \} \right) ^{d-k-2} \\
			&= 2\left(\max \{\frac{\sin \alpha}{\sin(\beta - \alpha)} , \sin \alpha \} \right) ^{d-k-2}\\
			& = 2 \left(\frac{\sin \alpha}{\sin(\beta - \alpha)}\right) ^{d-k-2}.
		\end{aligned}
	\end{multline*}\details{$\sin \arctan x = x/\sqrt{x^2+1}$}
	
	Therefore,
	\begin{equation}\label{equ:conditionalbound}
	\mathbb{P} \left( \abs{\hat{U} \cdot \hat{V}} \geq \cos \alpha \right) \leq (\sin \beta)^{d-k-1} +  2 \left(\frac{\sin \alpha}{\sin(\beta - \alpha)}\right) ^{d-k-2}.
	\end{equation}
	Note that we proved bound \eqref{equ:conditionalbound} conditioning on $A_i$'s, hence it is also a valid bound for random $A_i$'s (unconditionally).
	By \cref{lem:parallelogram}, \eqref{equ:conditionalbound} implies
	\begin{equation}
		\mathbb{P} \left(\mathcal{G}(\mathcal{B}_S \cap \mathcal{B}_T) \geq \frac{\epsilon^2 }{\sqrt{2\pi} \sin \alpha} \right) \leq   (\sin \beta)^{d-k-1} +  2 \left(\frac{\sin \alpha}{\sin(\beta - \alpha)}\right) ^{d-k-2}.
	\end{equation}
	Use inequalities $(2/\pi)x \leq \sin x \leq x$ for $0 \leq x \leq \pi/2$ to get
	\[
		\mathbb{P} \left(\mathcal{G}(\mathcal{B}_S \cap \mathcal{B}_T) \geq \frac{\epsilon^2 \sqrt{\pi} }{\sqrt{8} \alpha} \right) \leq \beta^{d-k-1} +  2 \left(\frac{ \pi \alpha}{2(\beta - \alpha)}\right) ^{d-k-2}.
	\]
	Set $\beta = \sqrt{\alpha}$ and restrict $0 < \alpha < 1/4$ so that $\sqrt{\alpha} \leq 1/2$. 
	The above probabilistic bound simplifies to
	\begin{align*}
		\mathbb{P} \left(\mathcal{G}(\mathcal{B}_S \cap \mathcal{B}_T) \geq \frac{\epsilon^2 \sqrt{\pi} }{\sqrt{8} \alpha} \right)
		& \leq \alpha^{(d-k-1)/2} +  2 \left(\frac{ \pi \alpha}{2(\sqrt{\alpha} - \alpha)}\right) ^{d-k-2}\\
		& = \alpha^{(d-k-1)/2} +  2 \left(\frac{ \pi \sqrt{\alpha}}{2(1 - \sqrt{\alpha})}\right) ^{d-k-2}\\
		& \leq \alpha^{(d-k-1)/2} +  2 \left( \pi \sqrt{\alpha} \right) ^{d-k-2}  \quad  (\text{use }1-\sqrt{\alpha} > 1/2) \\
		& \leq 3 \left( \pi \sqrt{\alpha} \right) ^{d-k-2}
	\end{align*}
	The claim follows by setting $\alpha = \frac{\pi}{2t}$.
\end{proof}

Combining \cref{lem:twobandsaff no k,lem:twobandsaff2} we get
\fi

\begin{lemma}\label{lem:twobandsaff}
	Let $A_1,\dots,A_k$, $S_1, \dots, S_{d-k}$, $T_1, \dots, T_{d-k}$ be iid.\ standard Gaussian random vectors in $\RR^d$ and $0 \leq k \leq d$.
	Let 
	\begin{align*}
	\mathcal{B}_S &= (\aff \{A_1, \dotsc, A_k, S_1, \dotsc, S_{d-k} \})_{\eps /2}, \\
	\mathcal{B}_T &= (\aff \{A_1, \dotsc, A_k, T_1, \dotsc, T_{d-k} \})_{\eps /2}.
	\end{align*}
	Then for $t > 2\pi$ we have
	\[
	\mathbb{P} \left(\mathcal{G}(\mathcal{B}_S \cap \mathcal{B}_T) \geq \frac{\epsilon^2  t}{\sqrt{2 \pi}} \right) 
	\leq 3 \left(\frac{ \pi ^{\hfrac{3}{2}} }{ \sqrt{2t} } \right)^{d-k-2}.
	\]
	
	\lnote{move pi right to left? $\sqrt{s}=\pi/\sqrt{t}$ to get 
	\[
		\mathbb{P} \left(\mathcal{G}(\mathcal{B}_S \cap \mathcal{B}_T) \geq \frac{ \pi^{5/2} }{\sqrt{8}}\epsilon^2 s \right) 
		\leq \frac{3}{s^{\hfrac{(d-k-2)}{2}}}.
	\]
	} 
	
\end{lemma}


Suppose $P_n = \conv(A_1, \dotsc, A_n)$ is a full-dimensional simplicial polytope in $\mathbb{R}^d$ and $\mathcal{F}_n$ is its set of facets.
For $S \in \mathcal{F}_n$, we abuse notation so that $S$ also denotes the index set of vertices of $S$.
Let $U_S$ be an unit inner normal vector of $\aff (A_S)$ to $P_n$.
Define $(\aff A_S)_{\eps^-} := \{x \in \RR^d \suchthat \ 0<d(x, \aff A_S) \leq \eps, U_S\cdot (x-A_s) \geq 0 , s \in S\}$.

\begin{lemma}\label{lem:bands2 lim}
	Let $\delta \in (0,1)$.
	Suppose $A_1, \dotsc, A_n$ are $d$-dimensional iid.\ standard Gaussian random vectors with $d = \floor{\delta n}$.
	Let $P_n = \conv(A_1, \dotsc, A_n)$, which is full-dimensional simplicial a.s.
	For $\eps > 0$,	define a.s.
	\[
	V_n = \mathcal{G}\left( \bigcup_{S \in \mathcal{F}_n } (\aff A_S)_{\eps^- } \setminus P_n \right).
	\]
	\begin{enumerate}
		\item $V_n \leq \frac{ \eps}{\sqrt{2\pi}}  \binom{n}{d} $.
		
		\item There exist $c_2, c_7, c_8 > 1$ (that depend only on $\delta$) such that when $\eps=\eps(d) \leq \hfrac{1}{(c_8c_2^d)}$ we have
		\(
		 \lim_{n \goesto \infty} \mathbb{P} \left(V_n \geq (\hfrac{c_2^d}{c_7}) \eps\right) = 1.
		\)		
	\end{enumerate}	
\end{lemma}

\begin{proof}[Proof of part 1]
	The upper bound follows from the union bound of at most $\binom{n}{d}$ facets and the fact that the 1-dimensional Gaussian density is upper bounded by $1/\sqrt{2\pi}$.
\end{proof}
\begin{proof}[Proof of part 2]
	From \cref{lem: num of facets}, there exists a constant $c_{\mathcal{F}} > 1$ (that depends only on $\delta$) such that $\mathbb{P}\left(\card{\mathcal{F}_n} \geq c_{\mathcal{F}}^d \right) \goesto 1 \text{ as }  n \goesto \infty.$
	Since $P_n$ is simplicial a.s., we may present $\mathcal{F}_n$ as a set of binary $n$-vectors with exactly $d$ ones.
	Let $A_{\mathcal{F}_n}(t)$ be the maximum number of vectors in $\mathcal{F}_n$ with pairwise Hamming distance greater than or equal to $t$.
	Similarly as in the proof of \cref{lem:gilbert general}, one can pick vectors greedily (Gilbert-Varshamov bound) so that when $\card{\mathcal{F}_n} \geq c_{\mathcal{F}}^d$ and $c\in (0,1)$, and using $n/d < 2/\delta$ when $d\geq 2$,
	\[
	A_{\mathcal{F}_n}(cd) \geq \frac{c_{\mathcal{F}}^d}{(ne/c d)^{cd}} 
	\geq \frac{c_{\mathcal{F}}^d}{(2e/c\delta)^{cd}}.
	\]
	\details{$\frac{n}{d} = \frac{n}{\floor{\delta n}} <  \frac{n}{\delta n - 1} \leq  \frac{2n}{\delta n}  = \frac{2}{\delta} $}
	Since $\lim_{c \to 0^+} (2e/c\delta)^{c} = 1$ and $(2e/c\delta)^{c}$ is increasing for $0 \leq c \leq \hfrac{2}{\delta}$, we can pick $c_1 \in (0,1)$ such that $(2e/c_1\delta)^{c_1} < c_{\mathcal{F}}$.
	Let $c_2 = \frac{ c_{\mathcal{F}} }{ \left(\hfrac{2e}{c_1\delta}\right)^{c_1} } > 1$.
	Then we have,
	\begin{align}
	\lim_{n \goesto \infty}\mathbb{P} \bigl( A_{\mathcal{F}_n}(c_1d) \geq c_2^d  \bigr) = 1. \label{eqn:exp subset of facets}
	\end{align}
	Here we get a subset of facets $\mathcal{T} \subseteq \mathcal{F}_n$ such that any two different facets in $\mathcal{T}$ share no more than $(1-\frac{c_1}{2}) d$ vertices, and $\card{\mathcal{T}} = c_2^{d}$ for some constants $0 < c_1 < 1$, $c_2 > 1$ (that depend only on $\delta$).
	Let $N = \card{\mathcal{T}}$.
	Let $\mathcal{B}_S = (\aff A_S)_{\eps^-}$, $S \in \mathcal{F}_n$.
	Using an argument similar to the proof of \cref{lem:union of bands(non-asym)}, we get
	\begin{align*}
	V_n 
	&= \mathcal{G}\left( \bigcup_{S \in \mathcal{F}_n } (\aff A_S)_{\eps^- } \setminus P_n \right) \\
	&= \mathcal{G}\left( \bigcup_{S \in \mathcal{F}_n } (\aff A_S)_{\eps^- }\right) -  \mathcal{G}\bigl( P_n \setminus (P_n)_{-\eps}\bigr) \nonumber\\
	&\geq \mathcal{G}\left(\bigcup_{S \in \mathcal{T}} (\aff A_S)_{\eps^-}\right) - \mathcal{G}\bigl( P_n \setminus (P_n)_{-\eps} \bigr) \nonumber\\
	&\geq \sum_{S \in \mathcal{T}} \mathcal{G}(\mathcal{B}_S) - \frac{1}{2} \sum_{S,T \in \mathcal{T}, S \neq T} \mathcal{G}(\mathcal{B}_S \cap \mathcal{B}_T) - \mathcal{G}\bigl( P_n \setminus (P_n)_{-\eps} \bigr).
	\end{align*}
	
	We are going to bound each of the three terms in the last expression.
	
	\paragraph{First term: $\sum_{S \in \mathcal{T}} \mathcal{G}(\mathcal{B}_S)$.}
	From \cref{lem:boundofaffmin}, there exists a constant $c_3 > 0$ (that depends only on $\delta$) such that
	\(
	\mathbb{P}\left(\max_{S \subseteq \left[n\right], |S| = d} \dist(\aff A_S, 0) \leq c_3  \right) \geq 1 - 2e^{- d}\).
	Moreover, we increase $c_3$ so that $c_3>1$, which ensures that $c_3 \geq \eps$.	
	Recall that $\mathcal{B}_S = (\aff A_S)_{\eps^-}$. 
	We get
	\begin{equation}
	\mathbb{P}\left(  \sum_{S \in \mathcal{T}} \mathcal{G}(\mathcal{B}_S) \geq \frac{N\eps }{\sqrt{2 \pi}} e^{- 2c_3^2}  \right) \geq 1 - 2e^{- d}. \label{eqn:V term1}
	\end{equation}
	\details{$\mathbb{P}\left(\max_{S \subseteq \left[n\right], |S| = d} \dist(\aff A_S, 0) \leq c_2 (2+\sqrt{n/d}) \right) \geq 1 - 2e^{- d}$ and $n/d \leq 2/\delta$}
	
	\paragraph{Second term: $\frac{1}{2}\sum_{S,T \in \mathcal{T}, S \neq T} \mathcal{G}(\mathcal{B}_S \cap \mathcal{B}_T)$.}
	Use \cref{lem:twobandsaff} in a union bound applied to all pairs of sets in $\mathcal{T}$.
	For $t > 2\pi$ we have,
	\[
	\frac{1}{2} \sum_{S,T \in \mathcal{T}, S \neq T} \mathcal{G}(\mathcal{B}_S \cap \mathcal{B}_T) 
	\leq \binom{N}{2} \frac{\epsilon^2  t}{\sqrt{2 \pi}}
	\] 
	holds with probability at least
	\begin{align*}
	1-3\binom{N}{2} \left(\frac{ \pi ^{\hfrac{3}{2}} }{ \sqrt{2t} } \right)^{d-(1-c_1/2) d-2}
	&\geq  1- \frac{3N^2 t^2}{\pi^3} \left( \frac{ \pi ^{\hfrac{3}{2}} }{ \sqrt{2t} } \right)^{c_1 d/2}  \nonumber \\
	&= 1 - \frac{3 t^2}{\pi^3} \left( c_2^2 \left( \frac{ \pi ^{\hfrac{3}{2}} }{ \sqrt{2t} }  \right)^{c_1/2} \right)^d .
	\end{align*}
	Choose $t = c_4 := \frac{1}{2} \pi^3 (ec_2^2)^{\frac{4}{c_1}}$ to get
	\begin{align}
	\mathbb{P} \left( \frac{1}{2} \sum_{S,T \in \mathcal{T}, S \neq T} \mathcal{G}(\mathcal{B}_S \cap \mathcal{B}_T) 
	\leq \binom{N}{2} \frac{\epsilon^2  c_4}{\sqrt{2 \pi}}  \right) \geq 1 - \frac{3 c_4^2}{\pi^3} e^{-d}. \label{eqn:V term2}
	\end{align}
	
	\paragraph{Third term: $\mathcal{G}\left(  P_n \setminus (P_n)_{-\eps} \right)$.}
	From \cref{lem: gauss neighborhood}, we know \( \mathcal{G}\bigl(P_n \setminus (P_n)_{-\eps} \bigr) \leq c_5\eps d^{1/4}\) for some absolute constant $c_5$.
	
	Combining \eqref{eqn:exp subset of facets}, \eqref{eqn:V term1} and \eqref{eqn:V term2} we conclude that with probability $1-o(1)$ as $d \goesto \infty$:
	\begin{align*}
	V_n 
	&\geq \sum_{S \in \mathcal{T}} \mathcal{G}(B_S) - \frac{1}{2} \sum_{S,T \in \mathcal{T}, S \neq T} \mathcal{G}(B_S \cap B_T) - \mathcal{G}\bigl( P_n \setminus (P_n)_{-\eps} \bigr)\\
	&\geq \frac{N\eps }{\sqrt{2 \pi}} e^{-2c_3^2}- \binom{N}{2} \frac{\epsilon^2  c_4}{\sqrt{2 \pi}} - c_5\eps d^{1/4} \\
	&\geq \frac{N \eps}{\sqrt{2\pi}} \left(e^{-2c_3^2} - \frac{N\eps c_4}{2} - \frac{\sqrt{2\pi} c_5 d^{1/4}}{N} \right).
	\end{align*}
	Note that $\hfrac{\sqrt{2\pi} c_5 d^{1/4}}{N}$ decays exponentially in $d$.
	Therefore, when $\eps \leq 1 / e^{2c_3^2}c_4N$,
	\[
	\lim_{n \goesto \infty} \mathbb{P} \left(V_n \geq \frac{ N\eps}{3\sqrt{2\pi} e^{2c_3^2} } \right) = 1.
	\]
	The proof is finished by setting $c_7 = 3\sqrt{2\pi} e^{2c_3^2}$ and $c_8 = e^{2c_3^2}c_4$.
\end{proof}

We are ready now to restate and prove the main result of the section.
\thmpolytopeconditioning*
\begin{proof}
	For $\diam(P_{n+1})$, by \cref{lem:laurent} we have
	\begin{align*}
	\prob \bigl(\diam(P_{n+1})^2 \leq 2d - 4\sqrt{dt} \bigr) 
	& = \prob \bigl( \norms{A_i - A_j} \leq 2d - 4\sqrt{dt}, \ \forall i \neq j \in [n+1]  \bigr) \\
	& \leq \prob \left( \bigcap_{i=1}^{\floor{(n+1)/2}} \norms{A_{2i-1} - A_{2i}} \leq 2d - 4\sqrt{dt}  \right) \\
	& \leq \left(  e^{-t} \right)^{n/2}. 
	\end{align*}
	We get the claimed bound by setting $t = d/16$.
	\details{	
		$
		\prob \left(\diam(P_n) \geq d \right)  \geq 1 - \left(  e^{-d/16} \right)^{\floor{n/2}} \geq 1 - e^{-\frac{nd}{64}}.
		$
	}

	Apply \cref{lem:bands2 lim} to $P_{n} = \conv(A_1, \dotsc, A_{n})$ with $\eps = \hfrac{1}{(c_8c_2^d)}$, we have
	\[
	\lim_{n \goesto \infty} \mathbb{P} \left(V_n \geq \frac{1}{c_7c_8} \right) = 1.
	\]
	Since $\vf(P_{n+1}) \leq \eps$ when $A_{n+1} \in V_n$, 
	\(
	\lim_{n \goesto \infty} \mathbb{P} \left( \vf(P_{n+1}) \leq 1/(c_8c_2^d) \right) \geq  \frac{1}{c_7c_8}.
	\)
	The claim follows by picking $c = 1/c_2$ and $c' = 1/c_7c_8$.
%
%
\end{proof}

\paragraph{Acknowledgments.}
We would like to thank Nina Amenta, Jes\'us De Loera, Miles Lopes, Javier Pe\~na, Thomas Strohmer, Roman Vershynin and Van Vu for helpful discussions. 
This material is based upon work supported by the National Science Foundation under Grants CCF-1657939, CCF-1422830, CCF-2006994 and CCF-1934568.

\bibliographystyle{alpha}
\bibliography{bib}

\ifstoc
\appendix

\section{Technical lemmas for \cref{sec:polytopeconditioning}}

\fi

\end{document}